\documentclass[a4paper,onecolumn,11pt,accepted=2026-10-12]{quantumarticle}
\pdfoutput=1
\usepackage[utf8]{inputenc}
\usepackage[T1]{fontenc}
\usepackage[english]{babel}
\usepackage[numbers,sort&compress]{natbib}

\usepackage{amsmath,amssymb,mathtools,amsthm}
\usepackage{graphicx}
\usepackage[colorlinks=true,linkcolor=blue,citecolor=blue,urlcolor=blue]{hyperref}
\usepackage{microtype}
\usepackage{subfig}
\usepackage{placeins}              

\renewcommand{\thesection}{\Roman{section}}
\renewcommand{\thesubsection}{\Alph{subsection}}
\renewcommand{\thesubsubsection}{\arabic{subsubsection}}

\usepackage{enumerate} 
\DeclareMathOperator{\Tr}{Tr}
\DeclareMathOperator{\Var}{Var}

\newtheorem{theorem}{Theorem}
\newtheorem{lemma}{Lemma}
\newtheorem{corollary}{Corollary}

\begin{document}
\normalfont
\raggedbottom

\title{Work statistics of sudden quantum quenches: A random matrix theory perspective on Gaussianity and its deviations}
\author{Miguel Tierz}\email{tierz@simis.cn}
\affiliation{Shanghai Institute for Mathematics and Interdisciplinary Sciences\\ Block A, International Innovation Plaza, No. 657 Songhu Road, Yangpu District,\\ Shanghai, China}
\date{October 13, 2026}

\begin{abstract}
We show that, for sudden quenches, the work distribution reduces to the statistics of the traces of powers of Haar unitaries, which are random unitary matrices drawn uniformly from the unitary group. For translation–invariant quadratic fermionic chains with interactions extending to $m$ neighbors and periodic boundary conditions, the Loschmidt amplitude admits a unitary matrix model / Toeplitz representation, which yields a work variable of the form $W=\sum_{r\le m} a_r\,\Re\Tr U^r$ (and in models with pairing terms -superconducting pairing- additional $b_r\,\Im\Tr U^r$ terms appear). By invoking multivariate central limit theorems for vectors of traces of unitaries, we obtain a Gaussian distribution for $P(W)$ with variance $\mathrm{Var}(W)=\tfrac12\sum_r r\,(a_r^2+b_r^2)$ and asymptotic independence across different powers. We also characterise the conditions under which non-Gaussian tails arise, for example from many interaction terms or their slow decay, as well as the appearance of Fisher–Hartwig singularities. We illustrate these mechanisms in the XY chain. Various numerical diagnostics support the analytical results.

\end{abstract}

\maketitle

\section{Introduction}

Work is a central stochastic quantity in nonequilibrium thermodynamics:
it quantifies not only the average energetic cost of a driving protocol
but also the size and structure of fluctuations around that cost. In the
quantum setting, the definition of work requires care because work is
not an observable associated with a single Hermitian operator. A widely
used and operationally motivated choice is the two-projective-measurement
(TPM) scheme~\cite{Talkner2007,Campisi2011,deffner2019quantum}, which
underlies quantum fluctuation relations and has been implemented via
interferometric and ancilla protocols in several
platforms~\cite{Dorner2013,Mazzola2013,Batalhao2014}.

In this paper we focus on \emph{sudden quenches} in isolated many-body
systems, where a control parameter is changed abruptly and the system
subsequently evolves unitarily under the post-quench Hamiltonian.
Sudden quenches are a standard theoretical and experimental workhorse:
they probe far-from-equilibrium physics on the shortest timescales,
connect directly to fidelity decay and orthogonality-catastrophe
phenomena~\cite{Gorin2006}, and play a key role in modern discussions
of dynamical quantum phase transitions (DQPTs), where nonanalytic
features emerge in real-time Loschmidt
amplitudes~\cite{Heyl2013,Heyl2018}.

\subsection*{Work statistics of quenches}

From the viewpoint of work statistics, quenches in extended systems
exhibit a characteristic separation of scales: the \emph{central} part
of the work distribution is typically close to Gaussian in generic
settings, while \emph{non-Gaussian} features become pronounced in
large-deviation regimes and especially near or at quantum critical
points. In such regimes, universal behavior can be understood via
quantum-to-classical mappings in a film geometry and the associated
critical Casimir physics, with the non-Gaussian large-deviation tail
governed by the boundary universality class set by the initial
state~\cite{GambassiSilva2011,GambassiSilva2012,Goold2018}. These
results hold broadly---well beyond quadratic fermionic models and
Gaussian initial states.

More recently, the role of interaction range has come into focus.
For long-range systems with power-law decaying couplings
$J_r\propto r^{-\alpha}$, non-Gaussian work statistics can arise even
away from quantum critical points in the strong long-range regime
$\alpha<d$~\cite{Gherardini2024,SolfanelliDefenu2025}, consistent with
the broader classification of long-range quantum
systems~\cite{Defenu2023}.

\subsection*{The random-matrix perspective}

Our results are complementary to this body of work: by recasting the
TPM work characteristic function as a unitary matrix model, we obtain
(i)~a \emph{rigorous, checkable criterion} for the Gaussian core
(the strong Szeg\H{o} condition on the Toeplitz symbol), (ii)~an
\emph{explicit variance formula} whose physical content is transparent,
and (iii)~a \emph{taxonomy} of the distinct mechanisms that produce
systematic deviations from Gaussianity---slow Fourier decay (long-range
couplings), Fisher--Hartwig singularities (criticality, boundary
effects), and proximity of Fisher zeros (DQPTs). The random-matrix
viewpoint naturally separates the Gaussian core (typical fluctuations,
controlled by multivariate CLTs for traces of Haar unitaries) from the
tail structure (controlled by Toeplitz asymptotics beyond strong
Szeg\H{o}), providing a mathematically sharp route to both.

\subsection*{TPM work and sudden quenches}

We consider a protocol in which external controls $\lambda(t)$ are
varied during a time window $t\in[0,\tau]$. The system Hamiltonian is
\begin{equation}
H(t)\equiv H(\lambda(t)),\qquad
H_i:=H(\lambda(0^-)),\quad H_f:=H(\lambda(\tau^+)).
\end{equation}
The drive generates the unitary
\begin{equation}
U_\tau=\mathcal T\exp\!\Big(-i\int_0^\tau H(\lambda(t))\,dt\Big).
\end{equation}
In the TPM scheme~\cite{deffner2019quantum}, one measures $H_i$ at
$t=0^-$, evolves with $U_\tau$, and measures $H_f$ at $t=\tau^+$.
If $H_i|n^0\rangle=E_n^{(0)}|n^0\rangle$ and
$H_f|m^1\rangle=E_m^{(1)}|m^1\rangle$, the work random variable is
\begin{equation}
w=E_m^{(1)}-E_n^{(0)},
\end{equation}
with probabilities $p_n^{(0)}=\langle n^0|\rho_0|n^0\rangle$ and
$p_{m|n}=|\langle m^1|U_\tau|n^0\rangle|^2$.
The characteristic function of work,
\begin{equation}
\chi_W(u)=\int_{\mathbb R}e^{iu w}P(w)\,dw,
\end{equation}
has the standard operator form
\begin{equation}\label{eq:chi-general}
\chi_W(u)=\Tr\!\Big(e^{iu H_f}\,U_\tau\,e^{-iu H_i}\,\rho_0\,U_\tau^\dagger\Big).
\end{equation}

In a \emph{sudden quench}~\cite{Gorin2006} the drive is instantaneous
at $t=0$: we set $\tau=0^+$ and evolve for a time $t$ under $H_f$.
If the prepared state is an eigenstate,
$H_i|\psi_0\rangle=E_0|\psi_0\rangle$ (the usual situation in our
applications: a translation-invariant ground state, or a domain-wall
preparation), then~\eqref{eq:chi-general} reduces to
\begin{equation}\label{eq:chi-G}
\chi_W(u)=e^{-iuE_0}\,\langle\psi_0|e^{iu H_f}|\psi_0\rangle
= e^{-iuE_0}\,G(-u),
\end{equation}
where
\begin{equation}
G(t):=\langle\psi_0|e^{-i H_f t}|\psi_0\rangle
\end{equation}
is the \emph{Loschmidt amplitude} and $L(t)=|G(t)|^2$ the Loschmidt
echo. Thus $P(w)$ is the Fourier transform of a Loschmidt-type object,
while $L(t)$ quantifies fidelity
decay~\cite{Talkner2007,Campisi2011,Esposito2009,Silva2008,Gorin2006,Heyl2013}.
This connection has driven much of the modern activity on DQPTs, where
nonanalytic features in $G(t)$ govern far-from-equilibrium
behavior~\cite{Heyl2013,Heyl2018}. On the experimental side,
interferometric protocols access $\chi_W(u)$ directly via an ancilla,
enabling reconstructions of $P(w)$ and tests of fluctuation relations
in a variety of platforms~\cite{Dorner2013,Mazzola2013,Batalhao2014}.

\subsection*{Why random matrices enter: Toeplitz determinants and trace
statistics}

The key structural observation we develop and exploit is that, for
sudden quenches in quadratic (or effectively free) spin/fermion models
with periodic boundary conditions (PBC), the Loschmidt amplitude admits
a Toeplitz (or block-Toeplitz) determinant representation whose symbol
is built from the post-quench single-particle dispersion
$\varepsilon(k)$~\cite{perez2024dynamical,perez2024hawking}. Toeplitz
determinants are exactly equivalent to partition functions of unitary
random-matrix ensembles via the Heine--Szeg\H{o}
identity~\cite{forrester2010log,garcia2020toeplitz}. In the present
context, this equivalence upgrades the work problem to a random-matrix
linear-statistics problem and allows us to import sharp probabilistic
tools.

In the settings studied here one has a Toeplitz representation
$G_N(t)=D_N(e^{-it\,\varepsilon})$, and therefore
$\chi_W(u)=e^{-iuE_0}\,D_N(e^{iu\varepsilon})$.
The deterministic phase from the zero Fourier mode only shifts the mean and
will often be suppressed; the nontrivial factor is a characteristic function
of a random-matrix linear statistic (Sec.~\ref{sec:fourier-inversion}).
Concretely, expanding the dispersion in Fourier modes,
\begin{equation}\label{modes}
\varepsilon(k)=a_0 + \sum_{m\ge 1}\big(a_m \cos mk + b_m \sin mk\big),
\end{equation}
the Heine--Szeg\H{o} identity recasts the characteristic function as an
average over the unitary group%
\footnote{We illustrate with finite range in formulas for brevity;
the Gaussian core extends to infinite-range symbols under
$\sum_{r\ge1} r(a_r^2+b_r^2)<\infty$, or to truncations with
$m(N)=o(N^{2/3})$; cf.\ Sec.~IV.},
For notational simplicity, we absorb the trivial phase from the zero
Fourier mode $a_0$ by writing
$D_N(e^{iu\varepsilon})=e^{iu N a_0}\,\chi(u)$; equivalently, $\chi(u)$ is the
characteristic function of the \emph{centered} linear statistic.

\begin{equation}\label{eq:chi-linear-stat}
\chi(u)=\mathbb{E}_{U(N)}\!\left[\exp\!\left(
iu\sum_{m\ge1}\bigl(\alpha_m \mathrm{Tr}\,U^m+\alpha_m^{\ast}\mathrm{Tr}\,U^{-m}\bigr)
\right)\right],
\end{equation}
with $\alpha_m=(a_m-ib_m)/2$ determined by the quench and microscopic
couplings. In other words, \emph{the logarithm of the characteristic
function is the cumulant generating function of a linear statistic of
traces of powers of a Haar-distributed unitary}.

\subsection*{Gaussian core and its breakdown}

This random-matrix viewpoint immediately clarifies when and why a
Gaussian work distribution emerges. For fixed $m$ and $N\to\infty$,
the traces of powers obey a multivariate central limit theorem: for
each $m$, $\Tr U^m$ converges to an independent complex Gaussian with
$\mathbb{E}[\Tr U^m]=0$ and $\Var(\Tr U^m)=m$, and different powers
become asymptotically
independent~\cite{DiaconisShahshahani1994,DiaconisEvans2001,johansson1997random}.
Recent multivariate refinements give super-exponential rates in total
variation and quantify regimes where the number of included harmonics
may grow with $N$ while retaining
Gaussianity~\cite{JohanssonLambert2021,DuitsJohansson2010}.

Two points are worth emphasizing for a physics audience. First, the
normalization differs from the textbook i.i.d.\ CLT: there is no
$1/\sqrt{N}$; instead the intrinsic variance scale of $\Tr U^m$ is
$m$, and Gaussianity emerges in the large-$N$ limit with $m$ fixed
(or slowly growing). Second, the physics of microscopic couplings
enters only through the weights $\alpha_m$ multiplying the
(asymptotically independent) Gaussian modes. As a consequence, to
leading order the work distribution $P(w)$ is Gaussian with variance
\begin{equation}\label{Var}
\Var[w]\;=\;2\,\sum_{m\ge 1} m\,|\alpha_m|^2
\;=\;\frac{1}{2}\sum_{m\ge 1} m\,(a_m^2+b_m^2),
\end{equation}
while controlled, model-specific non-Gaussian corrections arise when
at least one hypothesis behind the Szeg\H{o}/CLT regime is weakened.
We identify and illustrate the following mechanisms:
(A)~\emph{Too many modes / slow decay}---a large number of
harmonics or borderline Fourier decay (as in long-range couplings
with $J_r\propto r^{-\alpha}$, $\alpha\le 1$ in 1D) makes higher
cumulants visible;
(B)~\emph{Fisher--Hartwig singularities}---jumps, zeros, or
edge exponents in the Toeplitz
symbol~\cite{hartwig1969asymptotic,forrester2010log} generate
non-quadratic terms in the asymptotics and structured non-Gaussian
tails;
(C)~\emph{Critical quenches / DQPTs}---Fisher zeros of $G(t)$
approach the real-time axis, narrowing the analyticity strip of
$\chi_W(u)$ and thickening the tails of
$P(w)$~\cite{Heyl2013,Heyl2018};
(D)~\emph{Finite-$N$ effects}---at moderate $N$ or
with weight on medium/large $r$, asymptotically vanishing cumulants
remain visible.

\subsection*{Summary of main results}

\textit{(1) Mapping.}
For sudden quenches in translation-invariant quadratic fermionic chains
(PBC), the work variable maps to a CUE linear statistic:
\begin{equation*}
W=\sum_{m\ge 1}\big(a_m\,\Re\,\Tr\,U^m+b_m\,\Im\,\Tr\,U^m\big),\qquad
U\sim\text{Haar on }U(N),
\end{equation*}
where $\{a_m,b_m\}$ are the Fourier coefficients of the post-quench
dispersion~$\varepsilon(k)$. [Sec.~II--III]

\textit{(2) Gaussian core.}
Under the strong Szeg\H{o} condition
$\sum_{r\ge1} r(a_r^2+b_r^2)<\infty$, multivariate CLTs for traces
of Haar unitaries give $W\Rightarrow\mathcal{N}(0,\sigma^2)$ as
$N\to\infty$ with
\begin{equation*}
\sigma^2=\frac12\sum_{m\ge1} m(a_m^2+b_m^2)
\end{equation*}
and asymptotic independence across different powers. [Sec.~IV]

\textit{(3) Deviation mechanisms.}
Non-Gaussian features appear when:
(A)~too many modes or slow Fourier decay (long-range couplings);
(B)~Fisher--Hartwig singularities (criticality, boundaries);
(C)~Fisher zeros of $G(t)$ near the real axis (DQPTs);
(D)~finite-$N$ effects. [Sec.~V]

\textit{(4) Long-range threshold.}
For $J_r\propto r^{-\alpha}$: the Szeg\H{o} condition holds iff
$\alpha>1$ (1D), identifying the strong long-range regime
$\alpha\le 1$ as the regime where the strong Szeg\H{o}/finite-variance
CLT assumptions break down and pronounced deviations (non-Gaussianity
and/or non-standard scaling) are expected, consistent with
Refs.~\cite{Gherardini2024,SolfanelliDefenu2025}. [Sec.~V]

\subsection*{Scope}

Our focus is the TPM notion of work for sudden quenches, because in
this setting the exact reduction~\eqref{eq:chi-G} and the
Toeplitz/unitary-matrix-model mapping~\eqref{eq:chi-linear-stat}
provide a mathematically sharp route to Gaussianity criteria and their
breakdown. Extensions to finite-time driving~\cite{zawadzki2023non} or alternative notions of
quantum work~\cite{Talkner2007,Campisi2011} typically require
different structured-determinant objects (generically beyond Toeplitz
form) and are discussed in the outlook (Sec.~VII).

\subsection*{Organization}

In Sec.~II we set up the quadratic fermion models and the
Toeplitz/matrix-model representation of the Loschmidt amplitude for
translation-invariant Gaussian states and domain-wall preparations.
In Sec.~III we prove the generator identity and the resulting
identification of the work distribution with a random-matrix linear
statistic. In Sec.~IV we derive the Gaussian core using multivariate
CLTs for vectors of traces of Haar unitaries, including quantitative
regimes where the number of active harmonics may grow slowly with
system size. In Sec.~V we specialize to quadratic fermionic chains
(including pairing models) and discuss in detail when the Gaussian core
is guaranteed and how non-Gaussian tails arise, including a long-range
corollary and connections to entanglement via Fisher--Hartwig
singularities; we illustrate the mechanisms in the XY chain. Sec.~VI
presents numerical diagnostics (histograms, Q--Q plots, variance and
kurtosis scaling). We conclude in Sec.~VII with scope and outlook, and we also
include a brief discussion of physical implications of non-Gaussianity
at the end of Sec.~V.
\section{Setup}\label{sec:setup}

\noindent\textbf{Purpose.—}
This section fixes conventions for translation-invariant quadratic chains and the Gaussian-state
technology used throughout. These ingredients are the starting point for the Toeplitz/matrix-model
representations derived later, which ultimately enable the random-matrix formulation of the TPM work
characteristic function for sudden quenches.

\noindent\textbf{Model and conventions.—}
 We work with spinless fermions on a ring of length $L$ with periodic boundary
conditions (PBC). The lattice operators $c_j$ and $c_j^\dagger$ annihilate/create
a fermion at site $j$ and obey the canonical anti\-commutation relations (CAR)
\begin{equation}
\{c_i,c_j\}=0,\qquad \{c_i^\dagger,c_j^\dagger\}=0,\qquad \{c_i,c_j^\dagger\}=\delta_{ij}.
\end{equation}
With the discrete Brillouin set $\mathcal{B}=\{2\pi n/L:\,n=0,1,\dots,L\!-\!1\}$ we use
\begin{equation}
c_k:=\frac{1}{\sqrt{L}}\sum_{j=1}^{L} e^{-ik j}\,c_j,\qquad
c_j=\frac{1}{\sqrt{L}}\sum_{k\in\mathcal{B}} e^{ik j}\,c_k,
\end{equation}
so that $\{c_k,c_{k'}^\dagger\}=\delta_{k,k'}$ and other anti\-commutators vanish. We set $\hbar=1$ throughout. We consider finite–range, translation–invariant quadratic fermion chains on a ring of length $L$ with periodic boundary conditions (PBC). In real space,
\begin{align}
H &= \sum_{j=1}^{L}\sum_{r=1}^{m} \Big( J_r\, c^\dagger_{j+r} c_j + \overline{J_r}\, c^\dagger_j c_{j+r} \Big)
\nonumber\\
&\quad + \sum_{j=1}^{L}\sum_{r=1}^{m} \Big( \Delta_r\, c^\dagger_{j+r} c^\dagger_j + \overline{\Delta_r}\, c_j c_{j+r} \Big)
\;+\; \mu \sum_{j=1}^{L} c^\dagger_j c_j .
\label{eq:realH}
\end{align}
We call the number–conserving case $\Delta_r\equiv0$ \emph{XX‑type} (unpaired), and the case $\Delta_r\neq0$ \emph{paired/ Bogoliubov–de Gennes (BdG)}. For background on the free‑fermion/spin‑chain correspondence and diagonalization, see \cite{LiebSchultzMattis1961,Pfeuty1970,BarouchMcCoy1971}. For the BdG/Nambu formalism and Gaussian‑state technology, see \cite{Onishi1966,Robledo2009,klich2003}.

\medskip
\noindent\textit{Momentum‑space form.} After a Fourier transform (to $k=\tfrac{2\pi}{L}\mathbb{Z}$) and, when pairing is present, writing operators in the Nambu basis $\Psi_k=(c_k, c^\dagger_{-k})^T$, the Hamiltonian decouples into $2\times2$ blocks with Pauli matrices $(\sigma_x,\sigma_y,\sigma_z)$
\begin{equation}
H=\sum_{k}\Psi_k^\dagger H(k)\Psi_k,\qquad
H(k)=\xi(k)\,\sigma_z+\Delta(k)\,\sigma_y,
\label{eq:BdG}
\end{equation}
with single–particle dispersion $\varepsilon(k)=2\sqrt{\xi(k)^2+\Delta(k)^2}$. For finite range $m$,
\begin{align}
\xi(k) &= 2\sum_{r=1}^{m} \Re(J_r)\cos(rk) - \mu, \label{eq:xi}\\
\Delta(k) &= 2\sum_{r=1}^{m} \Re(\Delta_r)\sin(rk). \label{eq:Delta}
\end{align}
The unpaired case corresponds to $\Delta\equiv0$; pairing simply activates the sine harmonics through $\Delta(k)$.

\medskip
\noindent\textit{Initial states.} We use two standard Gaussian preparations: (i) a translation‑invariant Slater/BCS eigenstate of $H_i$ (typically its ground state); and (ii) a \emph{domain‑wall} Slater state occupying a contiguous block $A\subset\{1,\dots,L\}$,
\begin{align}
|{\rm DW}_A\rangle &= \prod_{j\in A} c^\dagger_j\,|0\rangle, \\
C_{ij} &= \langle{\rm DW}_A|\,c_i^\dagger c_j\,|{\rm DW}_A\rangle \;=\; \mathbf{1}_{\{i=j\in A\}} .
\end{align}
Under PBC, case (i) yields for the Loschmidt amplitude a (block‑)Toeplitz determinant with symbol built from $e^{-it\varepsilon(k)}$ (and the Bogoliubov angles in the paired case); see the Heine–Szeg\H{o} framework and its block extensions in \cite{forrester2010log,BasorEhrhardt2003,JinKorepin2004}. For the domain‑wall state, translation symmetry is broken and one obtains Toeplitz \emph{minors} \cite{garcia2020toeplitz} formed from the restricted propagator; see \cite{perez2014mapping,santilli2020phase,perez2024dynamical}. These reductions are the starting point for our matrix‑model/Toeplitz analysis of $\chi_W(u)$ developed below.

For a sudden quench from $H_i$ to $H_f$ with $|\psi_0\rangle$ an $H_i$ eigenstate, the work characteristic function equals the Loschmidt amplitude, $\chi_W(u)=\langle e^{iuW}\rangle=G(-u)$ and $P(W)=(2\pi)^{-1}\!\int e^{-iuW}G(u)\,du$. For the number–conserving (XX–type) case under PBC and a translation–invariant Slater/BCS initial state, $G_N(t)$ is a Toeplitz determinant with symbol $f_t(e^{i\theta})=e^{-it\varepsilon(\theta)}$; for a domain–wall Slater state, $G_N(t)$ becomes a Toeplitz \emph{minor} built from the restricted propagator (see \cite{perez2014mapping,santilli2020phase,perez2024dynamical}).
For range-$m$ hoppings $J_r$ the one-body dispersion is $\varepsilon(\theta)=2\sum_{r=1}^m J_r\cos(r\theta)$ and the symbol is $f_t(\theta)=\exp(-it\,\varepsilon(\theta))=\exp(-i\sum_{r=1}^m a_r\cos(r\theta))$ with
\begin{equation}\label{eq:arPBC}
a_r=2J_r\qquad \text{(PBC, XX normalization)}.
\end{equation}
Consequently the work random variable can be written as the linear statistic (that is, a weighted sum of $\cos r\theta_j$ and $\sin r\theta_j$ modes)
\begin{equation}\label{eq:Wlin}
W=\sum_{r=1}^m a_r\,\Re\,\mathrm{Tr}\,U^r,\qquad U\sim\text{Haar on }U(N).
\end{equation}
With pairing terms, as in the XY chain, the symbol has both $\cos(r\theta)$ and $\sin(r\theta)$ contributions and \eqref{eq:Wlin} generalizes to $W=\sum_{r\le m}\big(a_r\,\Re\,\mathrm{Tr}\,U^r+b_r\,\Im\,\mathrm{Tr}\,U^r\big)$. \emph{Pairing} here means quadratic Bogoliubov–de Gennes terms that do not conserve particle number—i.e., anomalous terms
$c_k c_{-k}$ and $c_{-k}^\dagger c_k^\dagger$—equivalently, a nonzero off–diagonal gap $\Delta(k)$ in the $2\times2$ BdG block. This setting will be studied further in Section V.

\paragraph*{Exact identification of the generator.}
Let $U\sim$ Haar on $U(N)$ with eigenangles $\{\theta_j\}_{j=1}^N$. For real vectors $(x_1,\dots,x_m)$ and $(y_1,\dots,y_m)$ define
\begin{equation}
V(\theta)=i\sum_{r=1}^m\big(x_r\cos r\theta+y_r\sin r\theta\big),\qquad f(\theta)=e^{V(\theta)}.
\end{equation}
\begin{lemma}[Exact generator identity (Heine--Szeg\H{o})]\label{lem:generator}
For any fixed $m$ and $N\in\mathbb N$,
\begin{equation}\label{eq:HSgenerator}
D_N(e^{V})=\mathbb{E}_{U(N)}\exp\!\Big(i\sum_{r=1}^m\big[x_r\,\Re\,\mathrm{Tr}\,U^r+y_r\,\Im\,\mathrm{Tr}\,U^r\big]\Big).
\end{equation}
\end{lemma}
\begin{proof}
The Heine--Szeg\H{o} identity~\cite{forrester2010log} states that for any symbol $f$ on the unit circle, $D_N(f)=\mathbb{E}_{U(N)}\prod_{j=1}^N f(e^{i\theta_j})$, where $\{\theta_j\}$ are the eigenangles of the Haar-distributed $U$. Applied to $f=e^{V}$, this gives $D_N(e^{V})=\mathbb{E}\exp\!\big(\sum_{j=1}^N V(\theta_j)\big)$. The result follows from the identities $\sum_j \cos r\theta_j=\Re\,\mathrm{Tr}\,U^r$ and $\sum_j \sin r\theta_j=\Im\,\mathrm{Tr}\,U^r$.
\end{proof}
\begin{corollary}[Work/Loschmidt generator]\label{cor:work}
With $\varepsilon(\theta)=2\sum_{r\le m}J_r\cos r\theta$ (and optional pairing $b_r\sin r\theta$), the Loschmidt amplitude is $G_N(t)=D_N(e^{-it\varepsilon})$. Up to the conventional sign in $\chi_W(u)=G(-u)$, the work characteristic function is the generator along the ray $(x_r,y_r)=(u\,a_r,\,u\,b_r)$ with $a_r=2J_r$:
\begin{equation}
\chi_W(u)=\mathbb{E}\exp\!\Big(iu\sum_{r\le m}\big[a_r\,\Re\,\mathrm{Tr}\,U^r+b_r\,\Im\,\mathrm{Tr}\,U^r\big]\Big).
\end{equation}
Thus the Fourier transform of the Toeplitz model for $G_N$ is exactly the joint moment  generator of $(\Re\,\mathrm{Tr}\,U^r,\Im\,\mathrm{Tr}\,U^r)_{r\le m}$ evaluated on a one-dimensional ray. In other words, $\chi_W(u)$ equals the moment-generating function of the multivariate random vector $(\Re\Tr U^r,\Im\Tr U^r)_{r\le m}$ evaluated along the specific line $(x_r,y_r) = (u a_r, u b_r)$ in its parameter space.
\end{corollary}

The physical couplings $J_r$ map to mode coefficients via \eqref{eq:arPBC}. Therefore,
\begin{equation}
\mathrm{Var}(W)=\frac12\sum_{r=1}^{m} r\,a_r^2=2\sum_{r=1}^{m} r\,J_r^2,
\end{equation}
and, with pairing coefficients $b_r$, $\mathrm{Var}(W)=\frac12\sum_r r\,(a_r^2+b_r^2)$. Different neighbor ranges contribute additively.

\section{Fourier transform of Loschmidt amplitude: \texorpdfstring{$P(W)$}{P(W)} as linear statistics of CUE eigenangles}\label{sec:fourier-inversion}

Physically, this section establishes that the work distribution $P(W)$ is the probability law of a specific random-matrix observable: a weighted sum of traces of powers of a Haar-distributed unitary, shifted by the initial energy. This is the identification that makes the entire random-matrix toolkit available for analyzing work statistics. The mathematical content is a precise statement (Theorem~\ref{thm:mainCF}) that the Toeplitz determinant $D_N(e^{iu\varepsilon})$ is the characteristic function of this observable, followed by Fourier inversion.

\paragraph*{Setup.}
Let $U\sim$ Haar on $U(N)$ with eigenangles $\{\theta_j\}_{j=1}^N$. Fix a real $2\pi$-periodic $\varepsilon$ with Fourier expansion
\begin{equation}\label{eq:eps-Fourier-ped}
\varepsilon(\theta)=\varepsilon_0+\sum_{r\ge 1}\big(\alpha_r\cos r\theta+\beta_r\sin r\theta\big),\qquad \alpha_r,\beta_r\in\mathbb R.
\end{equation}
Define the (real) linear statistic
\begin{equation}\label{eq:SN-def-ped}
S_N(U):=\sum_{j=1}^N \varepsilon(\theta_j)
= N\varepsilon_0+\sum_{r\ge1}\Big(\alpha_r\,\Re\,\mathrm{Tr}\,U^r+\beta_r\,\Im\,\mathrm{Tr}\,U^r\Big).
\end{equation}

Using Lemma~\ref{lem:generator}, we are now going to state and prove Theorem 1, which 
establishes that the work distribution $P(w)$ is the distribution of the random variable $S_N(U)-E_0$, i.e. a linear combination of random matrix traces offset by the initial energy. 

\begin{theorem}[Work characteristic function and law]\label{thm:mainCF}
\item\label{thm:CF} \emph{(i) Characteristic function.}
Let $|\psi_0\rangle$ be an eigenstate of $H_i$ with $H_i|\psi_0\rangle=E_0|\psi_0\rangle$ (so $E_0$ is the initial energy).
For a sudden quench to $H_f$ whose Loschmidt amplitude admits the Toeplitz representation
$G_N(t)=D_N(e^{-it\,\varepsilon})$ (this holds for translation-invariant free fermion models with PBC; see \cite{perez2024dynamical,perez2024hawking} and Section~V below), one has
\begin{equation}\label{eq:chiW-Toeplitz-ped}
\begin{split}
\chi_W(u)
&= \langle\psi_0|e^{iu H_f}e^{-iu H_i}|\psi_0\rangle \\
&= e^{-iuE_0}\,D_N\!\big(e^{iu\varepsilon}\big) \\
&= e^{-iuE_0}\,\mathbb{E}\,e^{\,iu\,S_N(U)}\,.
\end{split}
\end{equation}
So, using Lemma~1 with $(x_r,y_r)=(u a_r,\,u b_r)$ (Corollary~1), we identify $D_N(e^{iu\varepsilon})=\mathbb{E}[\exp(iu\,S_N(U))]$, hence $\chi_W(u)=e^{-iuE_0}\mathbb{E}[e^{iu\,S_N(U)}]$, proving (i).

\noindent\emph{(ii) Fourier inversion and pushforward law.}
Let $P_N$ denote the work law. As a tempered distribution on $\mathbb{R}$,
\begin{equation}
P_N(w)
= \frac{1}{2\pi}\int_{-\infty}^{\infty} e^{-i u w}\,\chi_W(u)\,du .
\label{eq:pushforward}
\end{equation}
Equivalently, for every Schwartz test function $\varphi\in\mathcal{S}(\mathbb{R})$,
\begin{equation}
\int_{\mathbb{R}} \varphi(w)\,P_N(dw)
= \mathbb{E}\!\left[\varphi\!\big(S_N(U)-E_0\big)\right].
\label{eq:pushforward-test}
\end{equation}
In particular, in Dirac notation,
$P_N=\mathbb{E}\big[\delta(\,\cdot-(S_N(U)-E_0))\big]$ (distributionally). This means that $P_N(w)$ is the probability law of $S_N(U)-E_0$ when $U$ is Haar-random (this is the meaning of \eqref{eq:pushforward-test}, interpreting the result as a (tempered) distribution).

\noindent\textit{Sketch of proof.}
Let $\widehat{\varphi}(u)\coloneqq\int_{\mathbb{R}} e^{-i u w}\,\varphi(w)\,dw$.
By part~(i), $\chi_W(u)=e^{-i u E_0}\,\mathbb{E}\!\left[e^{i u S_N(U)}\right]$.
Since $\widehat{\varphi}\in\mathcal{S}(\mathbb{R})\subset L^1(\mathbb{R})$, Tonelli/Fubini applies and
\begin{align}
\int \varphi(w)\,P_N(dw)
&= \frac{1}{2\pi}\int \widehat{\varphi}(u)\,\chi_W(u)\,du \nonumber\\
&= \frac{1}{2\pi}\int \widehat{\varphi}(u)\,e^{-i u E_0}\,
   \mathbb{E}\!\left[e^{i u S_N(U)}\right] du \nonumber\\
&= \mathbb{E}\!\left[\frac{1}{2\pi}\int \widehat{\varphi}(u)\,
   e^{i u (S_N(U)-E_0)} du\right] \nonumber\\
&= \mathbb{E}\!\left[\varphi\!\big(S_N(U)-E_0\big)\right],
\end{align}
which proves \eqref{eq:pushforward-test}. 
\end{theorem}

\paragraph*{Remarks.}
\noindent\emph{(1) What is and is not Fourier–transformed.}
By Lemma~1 and Theorem~1(i), $D_N(e^{iu\varepsilon})$ is (up to the trivial phase $e^{-iuE_0}$) the \emph{characteristic function} of the real linear statistic $S_N(U)$, see Eq.~(20); hence it is positive–definite and satisfies $\chi_W(0)=1$. The inverse Fourier transform acts on this characteristic function in the $u$–variable and, in the tempered sense, yields the pushforward law in Eq.~\eqref{eq:pushforward} (equivalently Eq.~\eqref{eq:pushforward-test}). In particular,
\[
P_N \;=\; \mathbb{E}\!\big[\delta\big(\,\cdot\,-(S_N(U)-E_0)\big)\big]
\quad\text{(as distributions)}.
\]
Crucially, nothing is Fourier–transformed pointwise in the angle $\theta$: the symbol $e^{iu\varepsilon(\theta)}$ \emph{enters through} the Toeplitz/matrix--model identity~\cite{forrester2010log,garcia2020toeplitz} that produces the characteristic function; the Fourier inversion is only in $u$.\\
(2) Note that each trace $\Tr U^r$ is complex, so writing $S_N$ in terms of $\Re \Tr U^r$ and $\Im \Tr U^r$ emphasizes that each mode contributes two independent Gaussian degrees of freedom. In other words, each trace is like a 2D Gaussian (pair of quadratures), ensuring $S_N$ is a real linear combination of those 2D components\\
(3) \emph{Centering and variance.} Since $\mathbb{E}[\mathrm{Tr}\,U^r]=0$, we have $\mathbb{E}[S_N]=N\varepsilon_0$ and
\begin{equation}\label{eq:VarSN}
\mathrm{Var}(S_N-N\varepsilon_0)=\frac12\sum_{r\ge1} r\big(\alpha_r^2+\beta_r^2\big)
\quad\text{(exact for $N\ge \max r$)},
\end{equation}
using $\mathbb{E}|\mathrm{Tr}\,U^r|^2=\min\{r,N\}$~\cite{DiaconisShahshahani1994} and orthogonality across different $r$ (see next Section for more on orthogonality). Notice that the real and imaginary parts of each $T_r$ are independent Gaussians with variance $r/2$ and this is why the variance \eqref{eq:VarSN} is a simple sum of contributions.

\begin{corollary}[Gaussian core with a growing number of modes]\label{cor:gaussian-growing-m}
Let $U\sim{\rm Haar}$ on $U(N)$ and set
\[
W_N=\sum_{r\ge 1}\big(a_r\,\Re\,\Tr U^r + b_r\,\Im\,\Tr U^r\big).
\]
Assume the Szeg\H{o} regularity $\sum_{r\ge1} r(a_r^2+b_r^2)<\infty$.
For any sequence $m(N)\to\infty$ with $m(N)=o(N^{2/3})$, write
\begin{equation}
W_N = W_N^{(\le m)} + R_N^{(>m)} ,\label{eq:WNsplit}
\end{equation}
with the truncation
\begin{equation}
\begin{aligned}
W_N^{(\le m)} &:= \sum_{r=1}^{m(N)} \big(a_r X_r + b_r Y_r\big),\\[-2pt]
&\text{with } X_r:=\Re\,\Tr U^r,\quad Y_r:=\Im\,\Tr U^r.
\end{aligned}
\label{eq:WNtrunc}
\end{equation}
Then, as $N\to\infty$,
\begin{equation}
W_N \ \Rightarrow\ \mathcal N\!\Big(0,\ \tfrac12\sum_{r\ge1} r(a_r^2+b_r^2)\Big).
\end{equation}
Moreover, for each compact $K\subset\mathbb{R}$,
\begin{equation}
\sup_{u\in K}\bigl|\chi_{W_N}(u)-e^{-u^2\sigma^2/2}\bigr|\ \longrightarrow\ 0.
\label{eq:char-unif}
\end{equation}
\[
\sigma^2=\tfrac12\sum_{r\ge1} r\,(a_r^2+b_r^2).
\]
In addition, if $m(N)\ll\sqrt{N}$ then the total-variation error for $W_N^{(\le m)}$
is super-exponentially small in $N/m(N)$ \cite{JohanssonLambert2021}.
\emph{Sketch.} By the multivariate CLT for traces \cite{JohanssonLambert2021}, the $2m(N)$–vector
$(X_1,Y_1,\dots,X_{m(N)},Y_{m(N)})$, scaled by $\sqrt{r/2}$ on each coordinate,
converges (in total variation) to a centered Gaussian in $\mathbb{R}^{2m(N)}$
whenever $m(N)=o(N^{2/3})$, with a super-exponential rate if $m\ll\sqrt{N}$.
Hence $W_N^{(\le m)}$ is asymptotically $\mathcal N(0,\tfrac12\sum_{r\le m(N)} r(a_r^2+b_r^2))$.
The tail has vanishing variance, $\Var R_N^{(>m)}=\tfrac12\sum_{r>m(N)} r(a_r^2+b_r^2)\to0$,
so \eqref{eq:char-unif} follows by Slutsky’s lemma.
\end{corollary}

\section{Gaussian core: why traces of powers are asymptotically independent Gaussians}\label{sec:CLT}

This section answers the central physical question: \emph{why is the work distribution generically Gaussian for large systems?} The answer is that the individual trace variables $\mathrm{Tr}\,U^r$ entering the work (Eq.~\eqref{eq:Wlin}) become independent Gaussians in the large-$N$ limit, with a variance that depends on $r$ but \emph{not} on $N$. The work, being a linear combination of these, inherits Gaussianity with the explicit variance $\frac12\sum_r r(a_r^2+b_r^2)$. We also explain why this CLT is qualitatively different from the textbook i.i.d.\ version.

\medskip
Let $U\sim$ Haar on $U(N)$ and write $T_r:=\mathrm{Tr}\,U^r=\sum_{j=1}^N e^{ir\theta_j}$. Already at finite $N$ one has \cite{DiaconisShahshahani1994,DiaconisEvans2001}
\begin{equation}\label{eq:2ndmoments}
\mathbb{E}\,|T_r|^2=\min\{r,N\},\qquad \mathbb{E}\,T_r=0,\qquad \mathbb{E}\,T_rT_s=0\ (r\neq s),
\end{equation}
This fixes the natural fluctuation scale: the variance of $T_r$ is of order $r$ (rather than $N$). For any fixed $m$, we have convergence in distribution to i.i.d. complex Gaussians, namely
\begin{equation}\label{eq:jointCLT-expanded}
\Big(\tfrac{T_1}{\sqrt{1}},\tfrac{T_2}{\sqrt{2}},\ldots,\tfrac{T_m}{\sqrt{m}}\Big)\;\xRightarrow[N\to\infty]{d}\;(Z_1,\ldots,Z_m),
\end{equation}
with $Z_r$ i.i.d.\ standard complex Gaussians. Equivalently,
\begin{equation}\label{eq:cov}
\mathbb{E}[T_r]=0,\quad \mathbb{E}[T_r\overline{T_s}]=\delta_{rs}\,r,\quad \mathbb{E}[T_r T_s]=0.
\end{equation}
So, different trace powers become independent Gaussians in the large-$N$ limit and it is the orthogonality $\mathbb{E}[T_r\overline{T_s}]=0$ for $r\neq s$ (and $\mathbb{E}[T_r T_s]=0$) that makes the joint distribution asymptotically factorize. There are two complementary derivations which are physically illuminating.

\emph{(i) Toeplitz/cumulants (strong Szeg\H{o}).} Heine--Szeg\H{o} identifies $D_N(e^{iu\varepsilon})$ with the characteristic function of $S_N=\sum_j\varepsilon(\theta_j)=\sum_r(\alpha_r \Re T_r+\beta_r \Im T_r)+N\varepsilon_0$. The strong Szeg\H{o} limit theorem expands $\log D_N$ into cumulants and shows that as $N\to\infty$ \cite{johansson1997random}
\begin{equation}
\log D_N(e^{iu\varepsilon})= iu\,N\varepsilon_0-\frac{u^2}{4}\sum_{r\ge1}r(\alpha_r^2+\beta_r^2) + o(1),
\end{equation}
so $\kappa_1(S_N)=N\varepsilon_0$, $\kappa_2(S_N)=\tfrac12\sum_r r(\alpha_r^2+\beta_r^2)$, and $\kappa_p(S_N)\to0$ for $p\ge3$. Recall that $o(1)$ denotes a term that vanishes as $N\to\infty$.


\emph{(ii) Direct multivariate CLT for traces.} K. Johansson \cite{johansson1997random} established the CLT for linear statistics, Rains~\cite{Rains1997} studied the distribution of $T_r$ when $r$ grows with $N$, and
Johansson and Lambert \cite{JohanssonLambert2021} gave quantitative multivariate bounds (TV metric) and growth windows for $m(N)$: the multivariate CLT remains valid uniformly for $m=o(N^{2/3})$ and when $m(N)=o(\sqrt{N})$ the total‑variation error is super‑exponentially small in $N/m$ \cite{JohanssonLambert2021}. This is convenient for controlling $\chi_W(u)$ uniformly on compact $u$-intervals.


Under any of these $m(N)$ growth values, for any fixed $m$ and coefficients $\{a_r,b_r\}_{r\le m}$,
\begin{equation}\label{eq:Wvar}
W=\sum_{r=1}^{m}\big(a_r\,\Re T_r+b_r\,\Im T_r\big)\ \Longrightarrow\ \mathcal{N}\!\Big(0,\; \tfrac12\sum_{r=1}^{m} r(a_r^2+b_r^2)\Big).
\end{equation}
This is our \emph{Gaussian core} for $P(W)$, verified below numerically. The Gaussian core is visible in the straight quantile–quantile plots, e.g. Fig.~\ref{fig:grid-3m-5m}(b,e) and Fig.~\ref{fig:grid-pl-exp}(b,e).
\paragraph*{Finite-$N$ corrections and growth of $m$.}
At fixed $m$, the leading non-Gaussian features are governed by the third and fourth cumulants.
When $m$ grows with $N$, the multivariate bounds of Johansson--Lambert quantify how fast $m$ may increase while $(T_1,\ldots,T_m)$ stays close (in total variation) to its Gaussian limit. Finite-$N$ deviations appear only in the extremes (detrended Q–Q): see Fig.~\ref{fig:grid-3m-5m}(c,f) and Fig.~\ref{fig:grid-pl-exp}(c,f). The surrogate and exact Haar data are indistinguishable in the bulk (compare Fig.~\ref{fig:grid-haar-surrogate}, rows 2–3).

Different powers $r\neq s$ are uncorrelated (indeed asymptotically independent), which justifies the variance additivity in \eqref{eq:VarSN} and later formulas.

\subsection{Products of traces: selection rules, Wick limits, and \texorpdfstring{large-$N$}{large-N} structure}\label{sec:Products}
Many observables with couplings to 1st, 2nd, \dots, $m$-th neighbors lead to correlators of products of traces such as $\mathbb{E}\big[\prod_{r\le m} T_r^{n_r}\,\overline{T_r}^{\bar n_r}\big]$. Two structural facts guide computations~\cite{DiaconisShahshahani1994,DiaconisEvans2001}.

\emph{Selection rule (exact for all $N$).} Invariance under the global phase $U\mapsto e^{i\phi}U$ implies
\begin{equation}
\mathbb{E}\!\Big[\prod_{a=1}^{p} T_{k_a}\;\prod_{b=1}^{q}\overline{T_{\ell_b}}\Big]=0
\quad\text{unless}\quad \sum_{a=1}^p k_a=\sum_{b=1}^q \ell_b.
\end{equation}
So, this invariance is why any moment with an unmatched number of $T$ factors vanishes. In particular $\mathbb{E}[T_1T_2\cdots T_m]=0$ for all $m\ge1$.

\emph{Wick factorization (asymptotic).} Since \eqref{eq:jointCLT-expanded} holds with independent complex Gaussian limits,
\begin{equation}
\mathbb{E}\!\Big[\prod_{r=1}^m T_r^{n_r}\,\overline{T_r}^{\bar n_r}\Big]
=\prod_{r=1}^m\big(\delta_{n_r,\bar n_r}\,n_r!\,r^{n_r}\big)+o(1).
\end{equation}
This means that mixed moment factorizes into second moments (pairings) plus $o(1)$. Thus, at leading order, only pairings $T_r$ with $\overline{T_r}$ survive; cross-pairings between different powers are subleading. Two basic consequences are, with $r,s\ \text{fixed},\ N\to\infty$
\begin{equation}
\mathbb{E}\,|T_r|^{2q}= q!\,r^{\,q}+o(1),\qquad
\mathbb{E}\big[T_r\,\overline{T_s}\big]=\delta_{rs}\,r.
\end{equation}

The lack of cross-covariance and the $\sqrt{r/2}$ scaling are borne out by the scatter panels with 95\% ellipses; see Fig.~\ref{fig:scatter-grid}(a,b) and the standardized versions in Fig.~\ref{fig:scatter-grid}(c,d).
\section{Quadratic fermionic chains, Gaussian core, and non-Gaussian tails}\label{sec:fermions}

\noindent\textbf{Purpose.—}
In this section we translate the abstract CUE\slash trace-CLT framework of Secs.~\mbox{III--IV} into concrete
criteria on the physical couplings of quadratic fermionic chains. We also catalogue the main
mechanisms that produce deviations from a finite-variance Gaussian work core (slow Fourier decay and
long-range couplings, Fisher--Hartwig singularities, and critical/DQPT effects), using the XY chain as
a representative benchmark.

\noindent\textbf{Notation.}
In the abstract framework of Secs.~\mbox{III--IV}, the Fourier coefficients of the dispersion were denoted $\alpha_r$ (cosine) and $\beta_r$ (sine). From this section onward we write $a_r$ and $b_r$ for the same quantities, absorbing any conventional factors of~2 (cf.\ Eq.~\eqref{eq:arPBC}). In particular, $\mathrm{Var}(W)=\tfrac12\sum_r r(a_r^2+b_r^2)$ and all formulas below use $a_r,b_r$ exclusively; the correspondence is $a_r\leftrightarrow\alpha_r$, $b_r\leftrightarrow\beta_r$.

Before turning to concrete models, it is worth emphasizing that the strong Szeg\H{o} condition \eqref{eq:HS12} is not just a technical regularity assumption. In the present setting it acts as a sharp \emph{physical diagnostic}: through the Fourier tail of the single-particle data it detects (i) long-range couplings (slow decay of $a_r,b_r$), (ii) critical or boundary-induced singular structure (Fisher--Hartwig data) that produces nonstandard fluctuation scaling, and (iii) parameter regimes where entanglement scaling is anomalous in free-fermion chains. In this sense, the Szeg\H{o}/Fisher--Hartwig classification provides a single language---and, in favorable cases, exact thresholds and variance formulas---for predicting when the work distribution has a finite-variance Gaussian core and when it does not.

\subsection{Starting point: a general quadratic chain (with pairing)}
Consider a translation-invariant quadratic fermionic Hamiltonian (with possible pairing) that, after Fourier and Nambu (BdG) lifting~\cite{LiebSchultzMattis1961,Onishi1966}, decouples as
\begin{equation}
H=\sum_{k\in\mathcal{B}} \Psi_k^\dagger\,\mathcal{H}(k)\,\Psi_k,\qquad
\mathcal{H}(k)=\xi(k)\,\sigma_z+\Delta(k)\,\sigma_y,
\end{equation}
with dispersion $\varepsilon(k)=2\sqrt{\xi(k)^2+\Delta(k)^2}$ and a Brillouin set $\mathcal{B}$. For a sudden quench from $H_i$ to $H_f$, the work characteristic function
\begin{equation}
\chi_W(u)=\langle e^{iu H_f}e^{-iu H_i}\rangle_{\psi_0}=e^{-iuE_0}\,\langle \psi_0|e^{iu H_f}|\psi_0\rangle
\end{equation}
is representation-invariant; it factorizes over $k$ and admits a compact determinant/Pfaffian form~\cite{klich2003,perez2024dynamical}. Under PBC and translation invariance, $\log\chi_W(u)$ is a Riemann sum over $k$ that, upon Fourier expansion in $\theta$, reduces to the Toeplitz model used above:
\begin{equation}\label{char}
\chi_W(u)=D_N\!\left(e^{i\sum_{r\ge 1}\,[x_r\cos(r\theta)+y_r\sin(r\theta)]}\right),
\end{equation}
with coefficients $(x_r,y_r)$ determined by the Fourier content of $\varepsilon$ together with the Bogoliubov angles that encode the initial BCS vacuum. Consequently the work variable is the linear statistic
\begin{equation}\label{WLS}
W=\sum_{r\ge1}\big(a_r\,\Re\,\mathrm{Tr}\,U^r+b_r\,\Im\,\mathrm{Tr}\,U^r\big),
\end{equation}
where $a_r\propto \widehat{\varepsilon}_r^{(\cos)}$ and $b_r\propto \widehat{\varepsilon}_r^{(\sin)}$ (pairing produces the sine-harmonics). We stress that in \eqref{char}, the coefficients $x_r$ and $y_r$ are directly related to the Fourier expansion of the single-particle energy $\varepsilon(k)$. Writing $\varepsilon(\theta) = \varepsilon_0 + \sum_{r\ge1} [\varepsilon^{(c)}_r \cos(r\theta) + \varepsilon^{(s)}_r \sin(r\theta)]$, and noting that pairing (nonzero $\Delta(k)$) contributes sine terms via the Bogoliubov rotation, one finds $x_r = u\,\varepsilon^{(c)}_r$ and $y_r = u\,\varepsilon^{(s)}_r$. Equivalently, $a_r = \varepsilon^{(c)}_r$ and $b_r = \varepsilon^{(s)}_r$ in the notation of \eqref{WLS} (up to the factor of $2$ coming from the Jordan-Wigner convention, cf. Corollary 1).\\
The fact that the initial state is a \emph{Bogoliubov} (BCS) vacuum rather than the Jordan--Wigner vacuum merely changes the values of $\{a_r,b_r\}$; the identities $\chi_W(u)=G(-u)$ and the Toeplitz representation are unaffected.

\subsection{When the Gaussian core is guaranteed}
Let $V(\theta)=iu\,\varepsilon(\theta)=iu\big(\varepsilon_0+\sum_{r\ge1}(\alpha_r\cos r\theta+\beta_r\sin r\theta)\big)$. The strong Szeg\H{o} condition
\begin{equation}\label{eq:HS12}
\sum_{r\ge1} r\,|\widehat{V}_r|^2=u^2\sum_{r\ge1} r(\alpha_r^2+\beta_r^2)<\infty \quad\;(V\in H^{1/2})
\end{equation}
is satisfied whenever $\varepsilon$ has Fourier modes decaying a bit faster than $1/r$~\cite{johansson1997random}. Under \eqref{eq:HS12}:
(i) the cumulant expansion gives $\kappa_1(S_N)=N\varepsilon_0$, $\kappa_2(S_N)=\tfrac12\sum_{r\ge1} r(\alpha_r^2+\beta_r^2)$ and $\kappa_p(S_N)\to 0$ for $p\ge 3$;
(ii) the multivariate CLT for $(T_1,\ldots,T_m)$ after $\sqrt{r}$ rescaling holds for fixed $m$, with quantitative extensions when $m=m(N)$ grows slowly.
Therefore, for $W=\sum_{r\le m}(a_r\,\Re T_r+b_r\,\Im T_r)$ one has
\begin{equation}
W\ \Rightarrow\ \mathcal{N}\!\Big(0,\;\tfrac12\sum_{r\le m} r(a_r^2+b_r^2)\Big)\qquad (N\to\infty).
\end{equation}
It is worth stressing that the strong Szeg\H{o} condition \eqref{eq:HS12} is what guarantees the Gaussian core rigorously: $W$ converges in distribution to $\mathcal{N}(0,\frac{1}{2}\sum r(a_r^2+b_r^2))$. 
Two useful corollaries: (1) if $\varepsilon$ is real-analytic (gapped models, generic BdG), then $(a_r,b_r)$ decay exponentially and the Gaussian core is extremely sharp; (2) even with mild cusps (e.g., linear dispersion near a gap closing), the sum in \eqref{eq:HS12} typically converges, so the core remains Gaussian while finite-$N$ corrections increase. The exponential vs.\ power-law Fourier decay manifests exactly as expected: the exponential case shows a virtually perfect Q–Q line, while the power-law case exhibits mild tail curvature (Fig.~\ref{fig:grid-pl-exp}(e) vs.\ Fig.~\ref{fig:grid-pl-exp}(b,c)).

\paragraph*{Long-range couplings and the Szeg\H{o} threshold.}

The strong Szeg\H{o} condition~\eqref{eq:HS12} acquires a transparent
physical interpretation when the couplings decay as a power law~\cite{Defenu2023}.
Consider a one-dimensional chain~\eqref{eq:realH} with hopping and
pairing amplitudes
\begin{equation}\label{eq:LR-couplings}
J_r = \frac{J}{r^{\alpha}}, \qquad
\Delta_r = \frac{\Delta}{r^{\alpha}},
\qquad r = 1,2,3,\ldots
\end{equation}
(with $J,\Delta$ setting the overall scales and $\alpha>0$ the decay
exponent). From Eqs.~\eqref{eq:xi}--\eqref{eq:Delta}, the dispersion
inherits Fourier coefficients that scale as $a_r \sim r^{-\alpha}$ and
$b_r \sim r^{-\alpha}$ (up to model-dependent constants and the
Bogoliubov rotation). Substituting into~\eqref{eq:HS12},
\begin{equation}\label{eq:LR-Szego-sum}
\sum_{r\ge 1} r\,(a_r^2 + b_r^2)
\;\sim\; \sum_{r\ge 1} r^{1-2\alpha},
\end{equation}
which converges if and only if $1-2\alpha < -1$, i.e.\
$\alpha > 1$. This identifies $\alpha = 1$ (in one spatial dimension)
as a sharp threshold within the Toeplitz/RMT framework:
\begin{itemize}
\item \emph{Weak long-range} ($\alpha > 1$, equivalently $\alpha > d$
in $d$ dimensions): the strong Szeg\H{o} condition is generically
satisfied, and the Gaussian core with variance
\mbox{$\frac{1}{2}\sum_r r(a_r^2+b_r^2) < \infty$} is robust. Physical
intuition: the effective number of strongly contributing harmonics is
finite, and the multivariate CLT for traces applies.
\item \emph{Strong long-range} ($\alpha \le 1$): the Szeg\H{o}
condition fails; the variance associated with the Gaussian-core formula
is no longer finite in the fixed-symbol setting, and the strong Szeg\H{o}/CLT
regime used above does not apply.
One expects pronounced deviations from the finite-variance Gaussian core
(non-Gaussianity and/or different scaling) even away from quantum critical points;
see Fig.~\ref{fig:longrange-threshold} and Fig.~\ref{fig:cue-variance} for numerical illustrations (a pedagogical long-range CLT proxy and a direct CUE variance check).
\end{itemize}
The threshold $\alpha = 1$ in one dimension coincides exactly with the
boundary $\alpha = d$ between the weak and strong long-range regimes in
the standard classification of long-range quantum
systems~\cite{Defenu2023}. It is
noteworthy that this boundary emerges here from a purely
\emph{mathematical} criterion (Szeg\H{o} regularity of the Toeplitz
symbol), rather than from a physical scaling or renormalization-group
argument. Yet the conclusion is fully consistent with recent independent
results obtained by very different methods: universal non-Gaussian work
statistics have been found in long-range systems at criticality for
$\alpha > d$~\cite{SolfanelliDefenu2025}, while
non-Gaussianity has been shown to persist even away from criticality in the strong
long-range regime $\alpha < d$~\cite{Gherardini2024}.

At the borderline $\alpha = 1$, the sum in~\eqref{eq:LR-Szego-sum}
diverges logarithmically: $\sum_{r=1}^{R} r^{-1} \sim \log R$.
Correspondingly, the variance of the work variable grows as
$\Var(W) \sim \log N$ when the number of active harmonics extends to
$r \sim N$, and the correct scaling for the centered work distribution
changes. This borderline regime connects to the Fisher--Hartwig
discussion in the next subsection: the Fourier coefficients
$c_m \sim m^{-1}$ that signal the onset of logarithmic variance growth
are precisely those associated with jump-type FH singularities in the
Toeplitz symbol (cf.\ the variance $\sim \sum_j \beta_j^2 \log N$
computed below for FH symbols with jump exponents $\beta_j$).

\paragraph*{Kac normalization and extensivity.}
The strong long-range regime $\alpha\le d$ is also the regime most relevant to current experimental
platforms for long-range quantum matter (e.g.\ trapped ions, Rydberg arrays, cavity-mediated systems),
where an $N$-dependent Kac normalization is routinely introduced to keep the energy extensive; see
Ref.~\cite{Defenu2023}. In this setting a natural question is whether Kac normalization restores the
\emph{finite-variance Gaussian-core} regime predicted by the strong Szeg\H{o} condition.

Implementing the Kac prescription amounts to replacing $J_r\to J_r/\mathcal{N}_\alpha(N)$ (and similarly
for pairing) with $\mathcal{N}_\alpha(N)=\sum_{r=1}^{N} r^{-\alpha}$. This makes the effective Fourier
weights $N$-dependent, $a_r\to a_r(N)$, and the Szeg\H{o} condition must be reconsidered in the joint
limit $N\to\infty$ with an evolving symbol.
For $\alpha<1$, $\mathcal{N}_\alpha(N)\sim N^{1-\alpha}/(1-\alpha)$, so that (heuristically)
$a_r(N)\sim r^{-\alpha}/N^{1-\alpha}$ and hence
\begin{equation}\label{eq:Kac-Szego}
\sum_{r=1}^{N} r\,a_r(N)^2
\;\sim\; \frac{1}{N^{2(1-\alpha)}}\sum_{r=1}^{N} r^{1-2\alpha}
\;=\;O(1).
\end{equation}
This estimate suggests that Kac normalization can restore the \emph{finiteness} of the Szeg\H{o} sum.
We therefore formulate the following working conjecture.

\emph{Conjecture (Gaussian core under Kac normalization).} For $\alpha<1$ with Kac normalization, the work
distribution recovers a Gaussian core with a finite limiting variance
\[
\sigma_\infty^2=\frac12\lim_{N\to\infty}\sum_{r=1}^{N} r\,a_r(N)^2<\infty,
\]
and higher cumulants vanish in the thermodynamic limit (after centering by the mean work).
A rigorous proof would require controlling Toeplitz determinants with an $N$-dependent symbol and
establishing CLT error bounds uniform in $N$; this goes beyond the standard strong Szeg\H{o} theorem~\cite{johansson1997random}
and its fixed-symbol extensions~\cite{JohanssonLambert2021}.

At the borderline $\alpha=1$, $\mathcal{N}_1(N)\sim\log N$, giving $a_r(N)\sim r^{-1}/\log N$ and
\[
\sum_{r=1}^{N} r\,a_r(N)^2 \sim (\log N)^{-2}\sum_{r=1}^{N} r^{-1}\sim \frac{1}{\log N}\to 0,
\]
suggesting that Kac normalization may ``over-correct'' the fluctuations at $\alpha=1$ and that a
different rescaling is required to obtain a nontrivial limiting work distribution in this borderline
case.

\subsection{Where non-Gaussian tails come from}
In several physically relevant mappings, such as for the case of open boundary conditions presented in \cite{perez2024dynamical} the symbol picks up genuine Fisher--Hartwig (FH) singularities such as algebraic endpoints and/or jump phases, unlike the PBC case studied here, where $\varepsilon(\theta)$ is smooth unless crossing a critical point. A canonical FH symbol has the form, with $\alpha_j>-1/2$ \cite{forrester2010log}
\begin{equation}
f_{\mathrm{FH}}(\theta)=e^{V(\theta)}\prod_{j=1}^{J}\big|2-2\cos(\theta-\theta_j)\big|^{\alpha_j}\,
e^{i\beta_j(\theta-\theta_j-\pi)}.\  
\end{equation}
For such symbols, the Toeplitz asymptotics depart from the pure strong Szeg\H{o} quadratic law: $D_N(f_{\mathrm{FH}})$ exhibits explicit $N$-dependent power laws, oscillatory factors, and constant terms determined by the FH data $\{\alpha_j,\beta_j,\theta_j\}$.\\ 
Under PBC, Loschmidt symbols naturally enter the Fisher--Hartwig (FH) class whenever the quench creates discontinuities at Fermi points or algebraic endpoints. For a \emph{flux (twist) quench} without pairing, one has a scalar Toeplitz symbol
$f_t(e^{i\theta})=\exp(i t\,\varepsilon(\theta))\,g(\theta)$ in which $g$ carries \emph{jump} FH singularities at the (pre/post) Fermi points $\{\theta_j\}$ with exponents $\beta_j\neq0$; expanding $\log f_t(\theta)=\sum_{m\neq0} c_m e^{im\theta}$ gives
$c_m^{(j)}=\tfrac{i\beta_j}{m} e^{-im\theta_j}$ and, by Heine--Szeg\H{o} plus the covariance $\mathbb{E}|\mathrm{Tr}\,U^m|^2=m$ \cite{DiaconisShahshahani1994,johansson1997random}, the linear statistic variance scales as
\[
\mathrm{Var}\big(\log G_N\big)=\sum_{m\le N} m |c_m|^2
=\Big(\sum_j \beta_j^2\Big)\log N + O(1),
\]
i.e.\ FH enhances fluctuations to a $\sqrt{\log N}$ scale while the \emph{Gaussian core} from the smooth part (Szeg\H{o}) remains.

For quenches with pairing (XY or Ising case, to be further discussed below), the object is a \emph{block Toeplitz} determinant with $2\times2$ symbol $\Phi(k,t)=(\mathbb{I}-\mathcal C_0(k))+\mathcal C_0(k)e^{-i\mathcal H_1(k)t}$; quenches across criticality generate \emph{root} FH exponents at $k=0,\pi$ ($\alpha=\pm\tfrac12$) and possibly jumps ($\beta\neq0$), producing deterministic power-law/oscillatory factors on top of the same CLT-driven Gaussian core for the finite Fourier harmonics of $\log\det\Phi$. Thus, in PBC as in the OBC case that will be discussed elsewhere, the trace-statistics mechanism gives the Gaussian center, while the FH data dictate the normalization and non-Gaussian tails.

Non-Gaussian features arise when at least one hypothesis behind Szeg\H{o}/CLT degrades. In particular:

\paragraph*{(A) Too many modes / slow decay.} If the effective number of harmonics $m$ grows with $N$ and $(a_r,b_r)$ decay slowly, higher cumulants need not vanish. A sufficient condition for a stable Gaussian core is $\sum_{r\ge1} r(a_r^2+b_r^2)<\infty$ together with admissible growth $m(N)$ (cf.~Johansson--Lambert). At the borderline $a_r,b_r\sim r^{-1}$ the variance diverges logarithmically and the correct scaling changes (non-Gaussian-looking histograms at fixed normalization). That is, a very large number of Fourier modes or a dispersion whose Fourier coefficients $a_r,b_r$ decay too slowly (borderline $1/r$ behavior or weaker) can invalidate the central limit approximation. See the tail curvature in detrended Q–Q, Fig.~\ref{fig:grid-pl-exp}(c), the kurtosis-scaling diagnostic in Fig.~\ref{fig:longrange-threshold}, and the direct CUE variance check in Fig.~\ref{fig:cue-variance}.

\paragraph*{(B) Fisher--Hartwig singularities.} In variants where the symbol develops genuine FH factors, which are nonanalytic features in the symbol such as jump discontinuities or zeros \cite{hartwig1969asymptotic}, $\log D_N$ contains non-quadratic terms (power-laws/oscillations), and persistent non-Gaussian edges appear. This includes certain OBC reductions and matrix-model weights related to Painlev\'e~V \cite{perez2024dynamical,BasorChenEhrhardt2010,LNR2018} (to be discussed elsewhere).

\paragraph*{(C) Critical quenches / DQPTs.} Crossing a quantum critical manifold brings Fisher zeros of $G(t)$ close to the real axis~\cite{Heyl2013,Heyl2018}. Indeed, if the quench crosses a phase transition, the Loschmidt amplitude develops Fisher zeros arbitrarily close to the real axis, which translates into $\chi_W(u)$ acquiring singularities for real $u$ (non-analytic behavior in time). This will thicken the tails of $P(W)$. Persistent edge features (cf. discussion here; qualitative signature matches the residuals in Fig.~\ref{fig:grid-3m-5m}(c,f)).

\paragraph*{(D) Finite-$N$ effects.} For moderate $N$ or weight on medium/large $r$, the third/fourth cumulants---which are $o(1)$ asymptotically---are visible. Edgeworth corrections from the cumulant method would capture these tails. Notice how small residual deviations are confined to extremes, Fig.~\ref{fig:grid-3m-5m}(c,f).


\subsubsection{Connection to entanglement: Fisher--Hartwig singularities
as a common origin}\label{subsec:entanglement-FH}

The Toeplitz/Fisher--Hartwig (FH) mechanisms behind deviations from the strong Szeg\H{o} regime
also underlie a familiar phenomenon in free-fermion many-body physics: \emph{anomalous scaling of
entanglement entropy}. Beyond a formal analogy, this yields a simple and experimentally testable
qualitative implication in quadratic chains:
\emph{parameter regimes that display anomalous entanglement scaling (e.g.\ area-law violations or
logarithmic growth) are natural candidates for persistent deviations from a finite-variance Gaussian
work core under sudden quenches, and vice versa}. The quantitative map between entanglement measures
and work cumulants is model dependent, but the \emph{trigger}---the regularity class of the underlying
Toeplitz symbols---is shared.

\paragraph{Entanglement and Toeplitz determinants.}
For the ground state of a translation-invariant free-fermion chain, the R\'enyi entanglement entropy
of a contiguous block of $\ell$ sites is determined by the eigenvalues of the $\ell\times\ell$
restricted correlation matrix $[C_\ell]_{jk}=\langle c_j^\dagger c_k\rangle$, $j,k\in A$.
Under PBC, $C_\ell$ is a Toeplitz matrix with symbol given by the momentum-space occupation function
$n(k)$. For paired/BdG chains the reduced covariance matrix is block-Toeplitz, but the same Toeplitz
asymptotic machinery applies.
The R\'enyi entropy $S_q(\ell)$ can be expressed in terms of a Toeplitz determinant $D_\ell(f_q)$
with a symbol $f_q$ built from $n(k)$ and the R\'enyi index $q$~\cite{JinKorepin2004}.
When $n(k)$ has jump discontinuities (as at Fermi points of a gapless chain), $f_q$ acquires FH
singularities. FH asymptotics then give the celebrated logarithmic scaling
$S_q(\ell)\sim c_q\log\ell$, with a prefactor fixed by the FH exponents~\cite{JinKorepin2004,Ares2015}.
In gapped short-range chains, $n(k)$ is smooth (Szeg\H{o}-regular) and entanglement obeys an area law.

\paragraph{Work statistics side: the same diagnostic.}
In our work problem, the determinant is different---it is the Toeplitz/block-Toeplitz determinant
for the Loschmidt amplitude (or $\chi_W$), with symbol built from the post-quench single-particle data
$\varepsilon_f(k)$ together with the Bogoliubov angles encoding the initial Gaussian state~\cite{perez2024dynamical,perez2024hawking}.
Nevertheless, the same classification governs the fluctuation regime:
\begin{itemize}
\item If the symbol is Szeg\H{o}-regular (Fourier coefficients decay fast enough), the multivariate
trace CLT yields a sharp Gaussian work core with finite variance
\mbox{$\Var(W)=\tfrac12\sum_r r(a_r^2+b_r^2)$}.
\item If the symbol develops FH structure (jumps/roots/edge exponents) or the Fourier tail is too slow
to satisfy \eqref{eq:HS12}, then the finite-variance CLT regime is lost and FH and finite-$N$ corrections
produce nonstandard scaling and visible deviations from a purely Gaussian description.
\end{itemize}
Thus the Szeg\H{o}/FH regularity class functions as a unified diagnostic for ``generic'' versus
``anomalous'' behavior in both entanglement and work statistics.

\paragraph{Concrete anchor: the XY chain (gapped vs critical).}
This parallel can be made vivid within the same model already used below.
For quenches within a gapped phase of the XY chain, the ground-state symbol entering $C_\ell$ is smooth
(area-law entanglement), while the post-quench dispersion and Bogoliubov data are analytic, leading to
exponentially decaying Fourier modes and an extremely sharp Gaussian work core (Sec.~\ref{subsec:XY-precise}).
By contrast, at (or across) a quantum critical line the correlation-matrix symbol acquires non-analytic
structure at $k=0,\pi$ that yields logarithmic entanglement~\cite{JinKorepin2004,Ares2015}, and the
Loschmidt block symbol develops FH-type root/jump singularities at the same special points, producing
logarithmic/power-law factors and enhanced non-Gaussian features in $P(W)$ (see the FH discussion in
Sec.~\ref{sec:fermions} and the XY analysis below).

\paragraph{Long-range chains: a bidirectional experimental prediction.}
The overlap is particularly concrete for power-law couplings $J_r\propto r^{-\alpha}$ in one dimension.
In this setting, the strong long-range regime $\alpha\le 1$ is known to generate algebraic
singularities in the correlation-matrix symbol and to produce area-law violations---including logarithmic,
fractal, and even volume-law scaling---even away from criticality
\cite{Ares2015b,Ares2018,Solfanelli2023}.
In our framework, the \emph{same regime} $\alpha\le 1$ is precisely where the Szeg\H{o} condition for
the work/Loschmidt symbol fails (Eq.~\eqref{eq:LR-Szego-sum}), i.e.\ where the finite-variance Gaussian-core
theorem no longer applies. We therefore expect that in free-fermion realizations of long-range chains,
an observation of anomalous entanglement scaling is accompanied by persistent deviations from a simple
Gaussian work distribution under sudden quenches, and conversely that pronounced non-Gaussian work features
signal an entanglement anomaly in the same parameter regime.

\paragraph{Scope and caveats.}
The symbols governing entanglement ($f_q[n(k)]$) and work ($e^{iu\varepsilon_f(k)}$ dressed by Bogoliubov
angles) are not identical, so FH exponents and prefactors need not match numerically. What our framework
does provide is a unified, sharp diagnostic: the Toeplitz regularity class of the underlying
single-particle data controls both entanglement and work statistics, and Toeplitz/FH asymptotics give a
single language in which quantitative relations (inequalities or functional links between $S_q(\ell)$ and
work cumulants) can be formulated and tested.

\subsubsection{The XY chain: square-root dispersion and tails}
\label{subsec:XY-precise}

\paragraph{Model and diagonalization.}
After the Jordan--Wigner map, the translation-invariant XY chain in a transverse field is a quadratic (BdG) fermion system~\cite{LiebSchultzMattis1961,Pfeuty1970,BarouchMcCoy1971} with
\begin{align}
\mathcal{H}(k) &= \mathbf{d}(k)\!\cdot\!\boldsymbol{\sigma},\\
\mathbf{d}(k) &= \big(\gamma\sin k,\,0,\,h-\cos k\big),\\
\varepsilon(k) &= 2\lvert\mathbf{d}(k)\rvert
= 2\sqrt{(h-\cos k)^2+(\gamma\sin k)^2}.
\end{align}
The \emph{Bogoliubov angle} is
\begin{equation}
\tan\!\big(2\theta_k\big)=\frac{\gamma\sin k}{\,h-\cos k\,},
\end{equation}
so that a rotation by $2\theta_k$ in Nambu space diagonalizes $\mathcal{H}(k)$.

\paragraph{Loschmidt amplitude and work characteristic function.}
For a sudden quench $(\gamma_i,h_i)\!\to\!(\gamma_f,h_f)$ from the ground state of $H_i$ (a BCS vacuum), the Loschmidt amplitude and TPM work characteristic function are
\begin{equation}
G(t):=\langle\psi_0|e^{-i H_f t}|\psi_0\rangle,
\qquad
\chi_W(u)=e^{-iuE_0}\,G(-u),
\end{equation}
with $H_i|\psi_0\rangle=E_0|\psi_0\rangle$~\cite{Silva2008}. The overlap factorizes over $(k,-k)$ and admits the standard product form~\cite{Heyl2013}
\begin{equation}
\label{eq:G-XY-product}
G(t)=\prod_{k>0}\Big[\,
\cos\!\big(\varepsilon_f(k)\,t\big)
-\, i\,\sin\!\big(\varepsilon_f(k)\,t\big)\,
\cos\!\big(2\Delta\theta_k\big)\,\Big],
\qquad
\end{equation} 
where $\Delta\theta_k:=\theta_k(\gamma_f,h_f)-\theta_k(\gamma_i,h_i)$.
\paragraph{From block Toeplitz to a scalar symbol.}
Equivalently, using the Klich reduction for Gaussian states~\cite{klich2003} (and the HFB/Pfaffian viewpoint~\cite{Onishi1966,Robledo2009}),
\begin{equation}
\Phi(k,t)=(\mathbb{I}-\mathcal C_0(k))+\mathcal C_0(k)\,e^{-i\mathcal H_f(k)t},
\end{equation}
with $\mathcal C_0(k)=\tfrac12(\mathbb{I}-\hat{\mathbf d}_i(k)\!\cdot\!\boldsymbol{\sigma})$.\\
For translation-invariant chains this is a \emph{block Toeplitz} operator in real space; see e.g.\ the block-Toeplitz structure for XY correlation-matrix functionals in~\cite{JinKorepin2004} and strong Szeg\H{o}--type expansions for matrix symbols in~\cite{BasorEhrhardt2003}. A short $2\times2$ calculus gives
\begin{equation}
\det \Phi(k,t)=\Big[\,
\cos\!\big(\varepsilon_f(k)t\big) 
-\, i\,\sin\!\big(\varepsilon_f(k)t\big)\,\cos\!\big(2\Delta\theta_k\big)\,\Big]^2.
\end{equation}
Thus the \emph{scalar} symbol
\begin{equation}
f_t(e^{ik})=\cos\!\big(\varepsilon_f(k)t\big)
- i\,\sin\!\big(\varepsilon_f(k)t\big)\,\cos\!\big(2\Delta\theta_k\big)
\end{equation}
controls $G(t)$ via $G(t)=\sqrt{\det T_N(\Phi)}=D_N(f_t)^{1/2}$ and \eqref{eq:G-XY-product}.

\paragraph{Small-$u$ reduction to traces.}
Setting $u=t$ in $\chi_W(u)=e^{-iuE_0}G(-u)$ and expanding for small $u$,
\begin{align}
f_u(e^{ik})
&= 1 - i\,u\,\tilde{\varepsilon}(k)
- \tfrac12 u^2\,\varepsilon_f(k)^2 + O(u^3),\\
\tilde{\varepsilon}(k)
&:= \varepsilon_f(k)\,\cos\!\big(2\Delta\theta_k\big),\\
\log f_u(e^{ik})
&= - i\,u\,\tilde{\varepsilon}(k) + O(u^2).
\end{align}
With the Fourier series
\begin{equation}
\tilde{\varepsilon}(k)=\sum_{m\in\mathbb Z}\tilde{\alpha}_m\,e^{i m k},
\qquad
\tilde{\alpha}_m=\frac{1}{2\pi}\!\int_{-\pi}^{\pi}\!
\tilde{\varepsilon}(k)\,e^{-i m k}\,dk,
\end{equation}
Heine--Szeg\H{o}~\cite{forrester2010log} gives the linear-statistics form
\begin{equation}
\log \chi_W(u)=
iu\sum_{m\ge 1}\Big(\tilde{\alpha}_m\,\Tr U^m
+\overline{\tilde{\alpha}_m}\,\Tr U^{-m}\Big) 
+\, O(u^2),
\end{equation}
so the selection rule and Wick factorization of Sec.~IV apply verbatim with
$\varepsilon(k)\mapsto\tilde{\varepsilon}(k)$, and the $a_r,b_r$ of our general framework are now the real and imaginary parts of $2\tilde{\alpha}_r$.
In particular,
\begin{equation}
\sigma^2 \;=\; 2\sum_{m\ge 1} m\,|\tilde{\alpha}_m|^2.
\end{equation}

\paragraph{When non-Gaussian tails appear.}
If the quench does \emph{not} cross a critical line, then $k\mapsto f_t(e^{ik})$ is smooth for fixed $t$; $\tilde{\varepsilon}(k)$ is analytic and the Fourier coefficients $\tilde{\alpha}_m$ decay exponentially, so the central $P(w)$ is Gaussian with $O(1)$ width. If the quench \emph{does} cross a critical line, $\Delta\theta_k$ acquires a endpoint jump at $k=0$ or $\pi$, $f_t$ becomes piecewise smooth (FH type), and Toeplitz FH asymptotics yield deterministic power-law/oscillatory corrections~\cite{DIK2011} (visible non-Gaussian tails).
\paragraph{XX limit}
For $\gamma_i=\gamma_f=0$ one has $\theta_k\equiv 0$, hence $\Delta\theta_k=0$ and \eqref{eq:G-XY-product} collapses to
\begin{equation}
    G(t)=\prod_{k>0} e^{-i\,\varepsilon_f(k)\,t},
\end{equation}
the translation-invariant, number-conserving XX limit (cf.\ \cite{LiebSchultzMattis1961,Pfeuty1970}). This should not be confused with the \emph{domain-wall XX} protocol analyzed earlier, where $G_N(t)=D_N(e^{-it\,\varepsilon})$ is a Toeplitz determinant built from the restricted propagator on a spatial block. Here $\varepsilon_f(k)=2\,|h_f-\cos k|$ is the BdG (positive) branch. 
If one prefers the Jordan–Wigner convention $\xi_f(k)=2(\cos k-h_f)$, the same reduction is obtained by keeping 
the piecewise sign in $\cos(2\Delta\theta_k)$, which amounts to writing $e^{-i\,|\xi_f|t}$ together with jump phases at the Fermi points (FH language). 
This change of convention does not affect the Gaussian-core variance nor our diagnostics.
\subsubsection{Quick summary of Gaussianity and tail behavior}
\begin{itemize}
\item \emph{Gapped BdG (analytic $\varepsilon$):} sharp Gaussian core; pairing just introduces $b_r$ in the variance $\tfrac12\sum r(a_r^2+b_r^2)$.
\item \emph{Critical (no FH):} core Gaussian; tails show visible but shrinking non-Gaussian corrections.
\item \emph{FH\slash jumps\slash zeros or borderline decay:} persistent non-Gaussian structures (edges\slash oscillations); the correct asymptotics are beyond strong Szeg\H{o}.
\item \emph{Large $m$ or slow decay of $(a_r,b_r)$:} the fixed-normalization histogram looks non-Gaussian; after the right rescaling a Gaussian limit often returns.
\end{itemize}


\subsubsection{Physical implications of non-Gaussianity}
\label{subsec:physical-implications}

The preceding analysis identifies when the work distribution is
Gaussian and when it is not, but one may ask: \emph{what are the
physical consequences of non-Gaussian work statistics?} We briefly
discuss three aspects that connect the Gaussian/non-Gaussian distinction
to observable thermodynamic behavior.

\paragraph{Fluctuation theorems and rare-event sampling.}
The Jarzynski equality and the Crooks fluctuation
theorem~\cite{Talkner2007,Campisi2011} relate the work distribution to
equilibrium free-energy differences and hold regardless of whether
$P(W)$ is Gaussian. However, the \emph{practical} use of these
relations---estimating $\Delta F$ from finitely many
realizations---depends sensitively on the tail structure of $P(W)$.
The Jarzynski estimator $\langle e^{-\beta W}\rangle$ is dominated by
rare trajectories in the left tail (negative-work events), and its
convergence is controlled by the weight in those tails. When $P(W)$ is
Gaussian, the tail probabilities are exponentially suppressed and the
estimator converges at a rate determined by the mean and variance alone.
Non-Gaussian tails---whether heavier (e.g., algebraic, as in FH
regimes) or lighter (e.g., bounded support)---change this convergence
behavior: heavy tails can make the estimator converge more slowly
(requiring exponentially more samples), while truncated tails can
improve it. The mechanisms identified in this work---slow Fourier decay,
FH singularities, Fisher zeros---provide a concrete, model-level
understanding of when and why such sampling difficulties arise in
sudden-quench protocols.

\paragraph{Irreversible work and higher cumulants.}
For a sudden quench from an eigenstate, the average work is
$\langle W\rangle = \langle\psi_0|H_f|\psi_0\rangle - E_0$. More
generally, when the initial state is thermal at inverse temperature
$\beta$, the irreversible (excess) work
$W_{\rm irr} = \langle W\rangle - \Delta F$ is related to the full
work distribution via the Jarzynski
equality~\cite{Talkner2007,Campisi2011}. When $P(W)$ is Gaussian,
$W_{\rm irr}$ is fully determined by the first two cumulants (mean and
variance); specifically, in the Gaussian thermal-state case,
$W_{\rm irr} = \beta\,\Var(W)/2$~\cite{Campisi2011}. When $P(W)$
develops non-Gaussian features, higher cumulants contribute and this
simple relation breaks down: the cumulant expansion of $W_{\rm irr}$
receives corrections from the third cumulant (skewness), fourth cumulant
(kurtosis), and beyond. Explicitly, expanding the Jarzynski equality in
cumulants gives
$W_{\rm irr}=\tfrac{\beta}{2}\kappa_2-\tfrac{\beta^2}{6}\kappa_3+\tfrac{\beta^3}{24}\kappa_4-\cdots$,
so the direction of the effect is set by the persistent cumulants:
positive skewness \emph{reduces} the excess work relative to the Gaussian
value, while positive excess kurtosis \emph{enhances} it. In the present framework, the cumulant structure
of $P(W)$ is directly controlled by the Toeplitz/RMT machinery: the
strong Szeg\H{o} theorem guarantees that $\kappa_p \to 0$ for $p \ge 3$
in the Gaussian regime, while Fisher--Hartwig corrections or slow
Fourier decay generate \emph{persistent} higher cumulants. Thus, the
same mechanisms that produce non-Gaussian tails also modify the relation
between work fluctuations and irreversibility, regardless of whether the
initial state is an eigenstate or a thermal mixture.

\paragraph{Fluctuation bounds and signal-to-noise.}
Variance-based precision--dissipation bounds provide constraints on how large work
fluctuations can be for a given amount of irreversibility. In particular, thermodynamic
uncertainty relations (TURs) bound the signal-to-noise ratio of currents in terms of entropy
production in classical and quantum Markovian settings~\cite{BaratoSeifert2015,HasegawaVanVu2019}.
For closed, unitary protocols and TPM work statistics the status of TUR-type bounds is more subtle,
but an important general message remains: second-moment information (mean/variance) captures the
full fluctuation structure only when $P(W)$ is close to Gaussian. When the work distribution has
persistent non-Gaussian features, higher cumulants carry independent information not captured by
variance-only bounds. The taxonomy developed here---distinguishing slow-decay, FH, and DQPT
mechanisms---thus provides a systematic way to identify quench protocols where rare events and
higher cumulants are expected to matter, and motivates cumulant-adapted bounds for such regimes.

\paragraph{Non-Gaussianity as a diagnostic.}
Finally, it is worth noting that the physical regimes producing
non-Gaussian work statistics---long-range interactions, proximity to
quantum critical points, Fisher--Hartwig singularities---are precisely
those where quantum correlations and many-body effects are most
pronounced. This suggests that the degree of non-Gaussianity in $P(W)$
(quantified, for instance, by the ratio of higher cumulants to
appropriate powers of the variance, or by the skewness proxy
$T_m/S_m^{3/2}$ defined in Sec.~VI) could serve as a practical,
experimentally accessible diagnostic for the complexity or
``quantumness'' of a quench protocol. Interferometric
reconstructions of $\chi_W(u)$~\cite{Dorner2013,Batalhao2014} already
provide access to the full characteristic function, and therefore in
principle to all cumulants; our results identify the microscopic
parameters (interaction range, symbol regularity) that control which
cumulants are visible.

\section{Numerical diagnostics}\label{sec:numerics}

The preceding sections make concrete predictions: a Gaussian work distribution with a specific variance formula, asymptotic independence of trace components, and identifiable non-Gaussian corrections in specific regimes. This section tests those predictions against Monte~Carlo sampling of the CUE linear statistic, using four complementary diagnostics: (i) histogram with Gaussian overlay; (ii) normal quantile-quantile plots (Q--Q plots; defined below) (and, when informative, detrended Q--Q); (iii) scatter plots of trace components with theory ellipses, including standardized versions; and (iv) variance and excess kurtosis versus $N$ with error bars. The full set of panels is shown after the references.

\paragraph*{Notation and basic definitions.}
Throughout this section, $N$ denotes the matrix dimension of $U(N)$ in the unitary matrix model; $m$ is the number of active Fourier harmonics in the dispersion (equivalently, the neighbor range retained in the chain); and $n$ is the number of Monte--Carlo realizations used per panel (sample size behind each histogram or Q--Q plot).
A \emph{Q--Q (quantile--quantile) plot} compares the empirical quantiles of the (centered) work variable on the vertical axis to the corresponding theoretical normal quantiles on the horizontal axis; alignment with the line $y=x$ indicates a Gaussian core, while systematic curvature near the extremes reveals tail deviations.
Histograms use the Freedman--Diaconis bin width \cite{FreedmanDiaconis1981}
\[
h \;=\; \frac{2\,\mathrm{IQR}}{n^{1/3}},
\quad
\mathrm{IQR} := Q_{0.75}-Q_{0.25},
\]
where $\mathrm{IQR}$ is the \emph{interquartile range} (the difference between the 75th and 25th percentiles).%



\paragraph*{Sampling and rendering choices.}
To display the Gaussian core with low sampling noise, most panels are generated with a surrogate. 

We use a Gaussian surrogate that preserves the exact first-/second-moment structure (Sec.~IV). To generate surrogate data, we replace the Haar matrix $U$ by independent complex Gaussian variables: $\Tr U^r$ is sampled from $\mathcal{CN}(0,r)$ (a complex normal with mean 0 and variance $r$) for each $r$, with different $r$ samples taken independent. This preserves the exact first- and second-moment structure but assumes asymptotic independence outright. This means that we draw:
\begin{equation}
\Tr U^r \stackrel{d}{\approx}\mathcal{CN}(0,r),
\end{equation}
independently over the finite set of harmonics and form $W=\sum_r (a_r\,\Re\,\Tr U^r + b_r\,\Im\,\Tr U^r)$
(cf.~\eqref{WLS}). This reproduces the central limit covariance exactly for fixed $r$ and large $N$, at a fraction of the cost of explicit Haar draws. As a cross-check, we include ``Haar vs surrogate'' comparisons at $N=80$, $n=300$, showing near-identical bulk behavior. See Fig.~\ref{fig:grid-haar-surrogate} for a Haar vs surrogate comparison.
Histograms use the Freedman--Diaconis bin width $h=2\,\mathrm{IQR}/n^{1/3}$ \cite{FreedmanDiaconis1981} and overlay the fitted Gaussian $\mathcal{N}(0,\hat\sigma^2)$.
Q--Q plots use the classic ``points + $y=x$'' style; we also provide \emph{detrended} Q--Q (residuals vs theoretical quantiles), which makes small tail deviations easier to see.

\paragraph*{What the panels show and how to read them.}
\begin{enumerate}
\item \textbf{Histogram/Q--Q pairs.} High-$n$ panels (histogram with Gaussian overlay; Q--Q with $y=x$ line) display the Gaussian core and isolate any tail curvature.
\item \textbf{Scatter plots.} Clouds of $(\Re\,\Tr U^r,\Re\,\Tr U^s)$ include the \emph{theory 95\% ellipse}; standardized versions divide each axis by $\sqrt{m/2}$ so the contour is a circle, making different $(r,s)$ pairs comparable.
\item \textbf{Variance vs $N$.} With fixed harmonics, the CLT variance $\mathrm{Var}(W)=\frac12\sum_r r(a_r^2+b_r^2)$ is $N$‑independent; the plot plateaus at this constant (dashed line) with SE bars $\hat\sigma^2\sqrt{2/(n-1)}$.
\item \textbf{Excess kurtosis vs $N$.} The trend toward $0$ (with error bars $\sqrt{24/n}$) matches the CLT/Wick picture; residual scatter reflects the estimator’s finite‑$n$ noise and small negative bias $\approx 6/(n+1)$. Panels~(c)--(d) in Fig.~\ref{fig:variance-kurtosis} illustrate the long‑range threshold: panel~(c) uses a long‑range correlated non‑Gaussian proxy field with $\langle y_0 y_r\rangle_c\sim r^{-\alpha}$, where the excess kurtosis decays toward $0$ for $\alpha>1$ but remains finite for $\alpha<1$; panel~(d) then checks directly in the CUE linear statistic that the variance follows the Szeg\H{o} sum and is finite for $\alpha>1$ but grows with $N$ for $\alpha<1$.
\end{enumerate}

\paragraph*{Compact diagnostics that tie figures to theory.}
For a given choice of weights, define
\[
S_m=\sum_{r\le m} r\,(a_r^2+b_r^2),\qquad
T_m=\sum_{r\le m} r^{3/2}\, (|a_r|^3+|b_r|^3).
\]
Then $\mathrm{Var}(W)=\tfrac12 S_m$ (prediction for the dashed line in Fig.~\ref{fig:var-vs-N}),
while the dimensionless skewness proxy $T_m/S_m^{3/2}$ tracks the relative size of the cubic cumulant: it is small (and shrinks with $m$) when the Fourier weights decay fast (analytic symbols, within-phase quenches), and grows in borderline/FH regimes. If all weight is on a single mode, this ratio corresponds exactly to the skewness of that mode's distribution. This explains why the Q--Q bulk is straight and why any curvature, when present, is confined to the extremes.

Figure~\ref{fig:variance-kurtosis}(b) shows how the kurtosis trends to the Gaussian value as $N$ increases, confirming our CLT-based prediction that $P(W)$ becomes Gaussian in the large-$N$ limit.

\paragraph*{Summary.}
Across the histogram/Q--Q pairs (Figures 1,2 and 3), and variance/kurtosis trends (Figure 4) and scatter panels (Figure 5) the data agree with the random-matrix CLT for traces: a \emph{Gaussian bulk} with the variance set by the quadratic form $\sum_r r(a_r^2+b_r^2)$ and \emph{asymptotic independence} across different powers. Non-Gaussian features, when visible, appear where Szeg\H{o} hypotheses are weakened (many harmonics, slow decay) or in settings with genuine FH structure (jumps/zeros/edge singularities), in line with the discussion in Sec.~V.


\section{Conclusions and outlook}

We mapped the work distribution in a sudden quench onto a random matrix problem (the traces of Haar unitaries). The multivariate central limit theorem then predicts a Gaussian core to the distribution for a large system size, $N$, with a variance fully determined by the quench parameters \eqref{Var}. More precisely, this is determined by the Fourier components of the post-quench dispersion. The coefficients $a_m$ and $b_m$ come from expanding the single-particle dispersion $\varepsilon(k)$ into cosine and sine harmonics, as given in equation \eqref{modes}. Note that, in contrast to summing $N$ independent random variables, the large-$N$ limit of the unitary group plays the role of the thermodynamic limit here, yielding an intrinsic variance of order $m$ rather than vanishing with $N$.

This is a direct physical consequence of the fact that the central–limit behavior we exploit here is of a
\emph{random–matrix} type rather than the textbook i.i.d.\ kind. For Haar $U(N)$,
the random variables $\Tr U^r$ have $O(1)$ variances (no $1/\sqrt{N}$ normalization), different
powers become asymptotically independent by unitary orthogonality, and the relevant
regularity condition is of Szeg\H{o} type \cite{johansson1997random} (summability of $r(a_r^2+b_r^2)$), not Lindeberg’s. These features lead to error controls and growth windows (e.g. total‑variation
CLTs up to $m(N)=o(N^{2/3})$ and super‑exponential accuracy if $m\ll\sqrt{N}$ \cite{JohanssonLambert2021}) that are quite unlike the Berry–Esseen picture for sums of i.i.d.\ variables; see also \cite{DiaconisShahshahani1994} for finite‑$N$ trace moments and orthogonality, \cite{DiaconisEvans2001} for the linear statistics CLTs, and the Gibbs lecture by Diaconis \cite{diaconis2003patterns} for context.

We identify how deviations (non-Gaussian tails) arise when the conditions of the CLT are weakened (multiple slow-decaying modes, symbol singularities, etc.), and we confirm these predictions with both analytic results (Toeplitz determinant asymptotics) and numerical simulations. This Gaussian ‘core’ of the work distribution, with model-specific non-Gaussian tails, suggests clear experimental signatures – e.g. the prevalence of Gaussian statistics in sudden quenches unless the system is tuned to special conditions (long-range interactions or critical points) that produce visible deviations.

Therefore, a natural next step is to drive the system across quantum critical manifolds (or to open boundaries) and ask how the \emph{tails} of the work distribution reorganize. In our mapping, the work characteristic function is the Loschmidt amplitude $\chi_W(u)=G(-u)$. When $G(t)$ develops \emph{Fisher zeros} that approach the real time axis, the strip of analyticity of $\chi_W(u)$ narrows, and by standard Fourier arguments this slows the decay of $P(W)$—revealing structured, non‑Gaussian tails. In parallel, when the (block)‑Toeplitz symbol acquires \emph{Fisher–Hartwig} features (jumps/zeros) at criticality or at hard edges (OBC), one expects robust algebraic/oscillatory fingerprints in the tails while the \emph{bulk} remains close to Gaussian. Charting these effects side‑by‑side—within‑phase vs across‑critical quenches—would provide an experimentally testable prediction from our Gaussian core with FH/DQPT-controlled tails picture. In particular, we already have computations on the open boundary case with the large $u$ (tails) asymptotics of the Toeplitz/Hankel determinant evaluated through Painlevé V: the expansion contains explicit power‑laws and oscillations that will give non‑Gaussian tails controlled by the nearest PV singularity to the real $u$ axis. This will be discussed elsewhere.

Likewise, by tuning the $\{a_r\}$ and probing the large‑deviation regime of the matrix model, one encounters the Gross--Witten--Wadia (GWW) transition~\cite{GrossWitten1980,Wadia1980}, a third-order phase transition in unitary matrix models whose realization in spin-chain Loschmidt amplitudes was identified in~\cite{perez2024dynamical}, together with its multicritical extensions~\cite{le2018multicritical}. These transitions reorganize the \emph{tails} of $P(W)$ in universal ways without changing the near‑Gaussian bulk. More broadly, the random-matrix perspective naturally suggests that the extreme-value statistics of the work distribution---i.e., the rare large-$|w|$ events that dominate Jarzynski sampling---should fall into Tracy--Widom universality classes~\cite{TracyWidom1994,forrester2010log}, with the specific class (soft edge, hard edge, or GWW-type) determined by whether the Toeplitz symbol has smooth endpoints, algebraic vanishing, or a critical saturation. Establishing this connection quantitatively would link quench thermodynamics to the well-charted landscape of random-matrix edge universality.

\paragraph*{Beyond sudden quenches and TPM work.}
The exact CUE/Toeplitz mapping developed here exploits the sudden-quench identity
$\chi_W(u)=e^{-iuE_0}G(-u)$ together with translation invariance, which yields a (block-)Toeplitz
determinant controlled by single-particle data. For more general \emph{finite-time} protocols in
quadratic systems, $\chi_W(u)$ can still be written in closed form as a determinant/Pfaffian of a
one-body operator built from the single-particle evolution (Klich-type formulas)
\cite{klich2003}, and recent work has shown that finite-time driving can lead to non-Gaussian work
statistics in its own right \cite{zawadzki2023non}. Alternative (non-TPM) notions of quantum work
correspond to different generating functions \cite{Talkner2007,Campisi2011} and are not addressed here.
For homogeneous drives the kernels in these determinant representations remain
Toeplitz/block-Toeplitz in real space, suggesting that Szeg\H{o}/Fisher--Hartwig regularity and
trace-CLT ideas may continue to control the emergence of a Gaussian core whenever the effective symbol
stays sufficiently regular, while protocols that generate non-analytic structure should produce
persistent higher-cumulant/tail effects.

\section*{Author contributions}
Miguel Tierz carried out the original research and wrote the manuscript.
Large language models (Claude and ChatGPT) were used during the revision
to edit parts of the text, check the bibliography and prepare the LaTeX files.
The author takes responsibility for the content of the work.

\section*{Acknowledgements}
We thank Gregory Schehr for correspondence.


\bibliographystyle{apsrev4-2-title}
\bibliography{refs_final}

@article{Talkner2007,
  author  = {Talkner, Peter and Lutz, Eric and H{\"a}nggi, Peter},
  title   = {Fluctuation theorems: Work is not an observable},
  journal = {Physical Review E},
  year    = {2007},
  volume  = {75},
  number  = {5},
  pages   = {050102(R)},
  doi     = {10.1103/PhysRevE.75.050102}
}

@article{Campisi2011,
  author  = {Campisi, Michele and H{\"a}nggi, Peter and Talkner, Peter},
  title   = {Colloquium: Quantum fluctuation relations: Foundations and applications},
  journal = {Reviews of Modern Physics},
  year    = {2011},
  volume  = {83},
  number  = {3},
  pages   = {771--791},
  doi     = {10.1103/RevModPhys.83.771}
}

@article{Esposito2009,
  author  = {Esposito, Massimiliano and Harbola, Upendra and Mukamel, Shaul},
  title   = {Nonequilibrium fluctuations, fluctuation theorems, and counting statistics in quantum systems},
  journal = {Reviews of Modern Physics},
  year    = {2009},
  volume  = {81},
  number  = {4},
  pages   = {1665--1702},
  doi     = {10.1103/RevModPhys.81.1665}
}

@article{Silva2008,
  author  = {Silva, Alessandro},
  title   = {Statistics of the work done on a quantum critical system by quenching a control parameter},
  journal = {Physical Review Letters},
  year    = {2008},
  volume  = {101},
  number  = {12},
  pages   = {120603},
  doi     = {10.1103/PhysRevLett.101.120603}
}

@article{Gorin2006,
  author  = {Gorin, Thomas and Prosen, Toma{\v z} and Seligman, Thomas H. and {\v Z}nidari{\v c}, Marko},
  title   = {Dynamics of {Loschmidt} echoes and fidelity decay},
  journal = {Physics Reports},
  year    = {2006},
  volume  = {435},
  number  = {2--5},
  pages   = {33--156},
  doi     = {10.1016/j.physrep.2006.09.003}
}

@article{Heyl2013,
  author  = {Heyl, Markus and Polkovnikov, Anatoli and Kehrein, Stefan},
  title   = {Dynamical Quantum Phase Transitions in the Transverse-Field {Ising} Model},
  journal = {Physical Review Letters},
  year    = {2013},
  volume  = {110},
  number  = {13},
  pages   = {135704},
  doi     = {10.1103/PhysRevLett.110.135704}
}

@article{Heyl2018,
  author  = {Heyl, Markus},
  title   = {Dynamical quantum phase transitions: a review},
  journal = {Reports on Progress in Physics},
  year    = {2018},
  volume  = {81},
  number  = {5},
  pages   = {054001},
  doi     = {10.1088/1361-6633/aaaf9a}
}

@article{Dorner2013,
  author  = {Dorner, Rodrigo and Clark, Stephen R. and Heaney, Liam and Fazio, Rosario and Goold, John and Vedral, Vlatko},
  title   = {Extracting Quantum Work Statistics and Fluctuation Theorems by Single-Qubit Interferometry},
  journal = {Physical Review Letters},
  year    = {2013},
  volume  = {110},
  number  = {23},
  pages   = {230601},
  doi     = {10.1103/PhysRevLett.110.230601}
}

@article{Mazzola2013,
  author  = {Mazzola, Laura and De Chiara, Gabriele and Paternostro, Mauro},
  title   = {Measuring the Characteristic Function of the Work Distribution},
  journal = {Physical Review Letters},
  year    = {2013},
  volume  = {110},
  number  = {23},
  pages   = {230602},
  doi     = {10.1103/PhysRevLett.110.230602}
}

@article{Batalhao2014,
  author  = {Batalh{\~a}o, Tiago B. and Souza, Alexandre M. and Mazzola, Laura and Auccaise, Ruben and Sarthour, Roberto S. and Oliveira, Ivan S. and Goold, John and De Chiara, Gabriele and Paternostro, Mauro and Serra, Roberto M.},
  title   = {Experimental Reconstruction of Work Distribution and Study of Fluctuation Relations in a Closed Quantum System},
  journal = {Physical Review Letters},
  year    = {2014},
  volume  = {113},
  number  = {14},
  pages   = {140601},
  doi     = {10.1103/PhysRevLett.113.140601}
}

@article{DiaconisShahshahani1994,
  author  = {Diaconis, Persi and Shahshahani, Mehrdad},
  title   = {On the eigenvalues of random matrices},
  journal = {Journal of Applied Probability},
  year    = {1994},
  volume  = {31A},
  pages   = {49--62},
  doi     = {10.1017/S0021900200106989}

}

@article{DiaconisEvans2001,
  author  = {Diaconis, Persi and Evans, Steven N.},
  title   = {Linear functionals of eigenvalues of random matrices},
  journal = {Transactions of the American Mathematical Society},
  year    = {2001},
  volume  = {353},
  number  = {7},
  pages   = {2615--2633},
  doi     = {10.1090/S0002-9947-01-02800-8}
}

@article{JohanssonLambert2021,
  author  = {Johansson, Kurt and Lambert, Gaultier},
  title   = {Multivariate normal approximation for traces of random unitary matrices},
  journal = {The Annals of Probability},
  year    = {2021},
  volume  = {49},
  number  = {6},
  pages   = {2961--3010},
  doi     = {10.1214/21-AOP1520}
}

@article{DuitsJohansson2010,
  author  = {Duits, Maurice and Johansson, Kurt},
  title   = {Powers of large random unitary matrices and {Toeplitz} determinants},
  journal = {Transactions of the American Mathematical Society},
  year    = {2010},
  volume  = {362},
  number  = {3},
  pages   = {1169--1187},
  doi     = {10.1090/S0002-9947-09-04542-5}
}

@article{Rains1997,
  author  = {Rains, Eric M.},
  title   = {High powers of random elements of compact {Lie} groups},
  journal = {Probability Theory and Related Fields},
  year    = {1997},
  volume  = {107},
  number  = {2},
  pages   = {219--241},
  doi     = {10.1007/s004400050084}
}

@article{BasorChenEhrhardt2010,
  author  = {Basor, Estelle L. and Chen, Yang and Ehrhardt, Torsten},
  title   = {{Painlev{\'e}~V} and time-dependent {Jacobi} polynomials},
  journal = {Journal of Physics A: Mathematical and Theoretical},
  year    = {2010},
  volume  = {43},
  number  = {1},
  pages   = {015204},
  doi     = {10.1088/1751-8113/43/1/015204}
}

@article{LNR2018,
  author  = {Lisovyy, Oleg and Nagoya, Hajime and Roussillon, Julien},
  title   = {Irregular conformal blocks and connection formulae for {Painlev{\'e}~V} functions},
  journal = {Journal of Mathematical Physics},
  year    = {2018},
  volume  = {59},
  number  = {9},
  pages   = {091409},
  doi     = {10.1063/1.5031841},
  eprint  = {1806.08344},
  archivePrefix = {arXiv}
}

@book{deffner2019quantum,
  title={Quantum Thermodynamics: An introduction to the thermodynamics of quantum information},
  author={Deffner, Sebastian and Campbell, Steve},
  year={2019},
  publisher={Morgan \& Claypool Publishers},
  doi     = {10.1088/2053-2571/ab21c6}

}

@article{perez2024dynamical,
  title={Dynamical quantum phase transitions from random matrix theory},
  author={P{\'e}rez-Garc{\'\i}a, David and Santilli, Leonardo and Tierz, Miguel},
  journal={Quantum},
  volume={8},
  pages={1271},
  year={2024},
  publisher={Verein zur F{\"o}rderung des Open Access Publizierens in den Quantenwissenschaften},
  doi     = {10.22331/q-2024-02-29-1271}

}

@article{perez2024hawking,
  title={{Hawking-Page} transition on a spin chain},
  author={P{\'e}rez-Garc{\'\i}a, David and Santilli, Leonardo and Tierz, Miguel},
  journal={Physical Review Research},
  volume={6},
  number={3},
  pages={033007},
  year={2024},
  publisher={APS},
  doi     = {10.1103/PhysRevResearch.6.033007}

}

@book{forrester2010log,
  title={Log-gases and random matrices ({LMS-34})},
  author={Forrester, Peter J},
  year={2010},
  publisher={Princeton university press},
  doi     = {10.1515/9781400835416}

}

@article{johansson1997random,
  title={On random matrices from the compact classical groups},
  author={Johansson, Kurt},
  journal={Annals of Mathematics},
  volume  = {145},
  number  = {3},
  pages={519--545},
  year={1997},
  publisher={JSTOR},
  doi     = {10.2307/2951843}
}

@article{hartwig1969asymptotic,
  title={Asymptotic behavior of {Toeplitz} matrices and determinants},
  author={Hartwig, Robert E and Fisher, Michael E},
  journal={Archive for Rational Mechanics and Analysis},
  volume={32},
  number={3},
  pages={190--225},
  year={1969},
  publisher={Springer},
  doi     = {10.1007/BF00247509}

}

@article{zawadzki2023non,
  title={Non-{Gaussian} work statistics at finite-time driving},
  author={Zawadzki, Krissia and Kiely, Anthony and Landi, Gabriel T and Campbell, Steve},
  journal={Physical Review A},
  volume={107},
  number={1},
  pages={012209},
  year={2023},
  publisher={APS},
  doi     = {10.1103/PhysRevA.107.012209}

}

@article{perez2014mapping,
  title={Mapping between the {Heisenberg} {XX} spin chain and low-energy {QCD}},
  author={Perez-Garcia, David and Tierz, Miguel},
  journal={Physical Review X},
  volume={4},
  number={2},
  pages={021050},
  year={2014},
  publisher={APS},
  doi     = {10.1103/PhysRevX.4.021050}

}

@article{santilli2020phase,
  title={Phase transition in complex-time {Loschmidt} echo of short and long range spin chain},
  author={Santilli, Leonardo and Tierz, Miguel},
  journal={Journal of Statistical Mechanics: Theory and Experiment},
  volume={2020},
  number={6},
  pages={063102},
  year={2020},
  publisher={IOP Publishing},
  doi     = {10.1088/1742-5468/ab837b}

}

@article{LiebSchultzMattis1961,
  author  = {Lieb, Elliott and Schultz, Theodore and Mattis, Daniel},
  title   = {Two soluble models of an antiferromagnetic chain},
  journal = {Annals of Physics},
  year    = {1961},
  volume  = {16},
  pages   = {407--466},
  doi     = {10.1016/0003-4916(61)90115-4}
}

@article{Pfeuty1970,
  author  = {Pfeuty, Pierre},
  title   = {The one-dimensional {Ising} model with a transverse field},
  journal = {Annals of Physics},
  year    = {1970},
  volume  = {57},
  pages   = {79--90},
  doi     = {10.1016/0003-4916(70)90270-8}
}

@article{BarouchMcCoy1971,
  author  = {Barouch, Eytan and McCoy, Barry M.},
  title   = {Statistical Mechanics of the {XY} Model. {II}. Spin-Correlation Functions},
  journal = {Physical Review A},
  year    = {1971},
  volume  = {3},
  number  = {2},
  pages   = {786--804},
  doi     = {10.1103/PhysRevA.3.786}
}

@article{klich2003,
  title={An elementary derivation of {Levitov}'s formula},
  author={Klich, Israel},
  journal={Quantum Noise in Mesoscopic Physics},
  pages={397--402},
  year={2003},
  publisher={Springer},
  doi={10.1007/978-94-010-0089-5_19}
}

@article{Onishi1966,
  author  = {Onishi, Naoki and Yoshida, Shiro},
  title   = {Generator coordinate method applied to nuclei in the transition region},
  journal = {Nuclear Physics},
  year    = {1966},
  volume  = {80},
  number  = {2},
  pages   = {367--376},
  doi     = {10.1016/0029-5582(66)90096-4}
}

@article{Robledo2009,
  author  = {Robledo, L. M.},
  title   = {Sign of the overlap of {Hartree--Fock--Bogoliubov} wave functions},
  journal = {Physical Review C},
  year    = {2009},
  volume  = {79},
  pages   = {021302},
  doi     = {10.1103/PhysRevC.79.021302}
}

@article{JinKorepin2004,
  author  = {Jin, B.-Q. and Korepin, V. E.},
  title   = {Quantum Spin Chain, {Toeplitz} Determinants and the {Fisher--Hartwig} Conjecture},
  journal = {Journal of Statistical Physics},
  year    = {2004},
  volume  = {116},
  pages   = {79--95},
  doi     = {10.1023/B:JOSS.0000037230.37166.42}
}

@article{DIK2011,
  title={Asymptotics of {Toeplitz}, {Hankel}, and {Toeplitz+Hankel} determinants with {Fisher-Hartwig} singularities},
  author={Deift, Percy and Its, Alexander and Krasovsky, Igor},
  journal={Annals of mathematics},
  volume={174},
  number={2},
  pages={1243--1299},
  year={2011},
  publisher={JSTOR},
  doi     = {10.4007/annals.2011.174.2.12}

}

@article{BasorEhrhardt2003,
  author  = {Basor, Estelle L. and Ehrhardt, Torsten},
  title   = {Asymptotics of Determinants of {Bessel} Operators},
  journal = {Communications in Mathematical Physics},
  year    = {2003},
  volume  = {234},
  number  = {3},
  pages   = {491--517},
  doi     = {10.1007/s00220-002-0769-1}
}

@article{FreedmanDiaconis1981,
  author  = {Freedman, David and Diaconis, Persi},
  title   = {On the histogram as a density estimator: {$L_{2}$} theory},
  journal = {Zeitschrift f{\"u}r Wahrscheinlichkeitstheorie und Verwandte Gebiete},
  year    = {1981},
  volume  = {57},
  pages   = {453--476},
  doi     = {10.1007/BF01025868}
}

@article{garcia2020toeplitz,
  title={{Toeplitz} minors and specializations of skew {Schur} polynomials},
  author={Garc{\'\i}a-Garc{\'\i}a, David and Tierz, Miguel},
  journal={Journal of Combinatorial Theory, Series A},
  volume={172},
  pages={105201},
  year={2020},
  publisher={Elsevier},
  doi={10.1016/j.jcta.2019.105201}
}

@article{le2018multicritical,
  title={Multicritical edge statistics for the momenta of fermions in nonharmonic traps},
  author={Le Doussal, Pierre and Majumdar, Satya N and Schehr, Gr{\'e}gory},
  journal={Physical review letters},
  volume={121},
  number={3},
  pages={030603},
  year={2018},
  publisher={APS},
  doi={10.1103/PhysRevLett.121.030603}
}

@article{diaconis2003patterns,
  title={Patterns in eigenvalues: the 70th {Josiah Willard Gibbs} lecture},
  author={Diaconis, Persi},
  journal={Bulletin of the American Mathematical Society},
  volume={40},
  number={2},
  pages={155--178},
  year={2003},
  doi={10.1090/S0273-0979-03-00975-3}
}

@misc{GambassiSilva2011,
  author        = {Gambassi, Andrea and Silva, Alessandro},
  title         = {Statistics of the Work in Quantum Quenches, Universality and the Critical {Casimir} Effect},
  year          = {2011},
  eprint        = {1106.2671},
  archivePrefix = {arXiv},
  primaryClass  = {cond-mat.stat-mech},
  url           = {https://arxiv.org/abs/1106.2671},
  doi           = {10.48550/arXiv.1106.2671}
}

@article{GambassiSilva2012,
  author  = {Gambassi, Andrea and Silva, Alessandro},
  title   = {Large Deviations and Universality in Quantum Quenches},
  journal = {Phys. Rev. Lett.},
  volume  = {109},
  pages   = {250602},
  year    = {2012},
  doi     = {10.1103/PhysRevLett.109.250602}
}

@incollection{Goold2018,
  author    = {Goold, John and Plastina, Francesco and Gambassi, Andrea and Silva, Alessandro},
  title     = {The Role of Quantum Work Statistics in Many-Body Physics},
  booktitle = {Thermodynamics in the Quantum Regime: Fundamental Aspects and New Directions},
  editor    = {Binder, Felix and Correa, Luis A. and Gogolin, Christian and Anders, Janet and Adesso, Gerardo},
  publisher = {Springer},
  address   = {Cham},
  pages     = {317--336},
  year      = {2019},
  doi       = {10.1007/978-3-319-99046-0_13},
  eprint    = {1804.02805},
  archivePrefix = {arXiv},
  primaryClass  = {cond-mat.stat-mech}
}

@article{Defenu2023,
  author  = {Defenu, Nicol{\`o} and Donner, Tobias and Macr{\`i}, Tommaso and Pagano, Guido and Ruffo, Stefano and Trombettoni, Andrea},
  title   = {Long-Range Interacting Quantum Systems},
  journal = {Rev. Mod. Phys.},
  volume  = {95},
  pages   = {035002},
  year    = {2023},
  doi     = {10.1103/RevModPhys.95.035002}
}

@article{SolfanelliDefenu2025,
  author  = {Solfanelli, Andrea and Defenu, Nicol{\`o}},
  title   = {Universal Work Statistics in Long-Range Interacting Quantum Systems},
  journal = {Phys. Rev. Lett.},
  volume  = {134},
  pages   = {030402},
  year    = {2025},
  doi     = {10.1103/PhysRevLett.134.030402}
}

@article{Gherardini2024,
  author  = {Gherardini, Stefano and Buffoni, Lorenzo and Defenu, Nicol{\`o}},
  title   = {Universal Defects Statistics with Strong Long-Range Interactions},
  journal = {Phys. Rev. Lett.},
  volume  = {133},
  pages   = {113401},
  year    = {2024},
  doi     = {10.1103/PhysRevLett.133.113401}
}

@article{Ares2015,
  author  = {Ares, Filiberto and Esteve, Jos{\'e} G. and Falceto, Fernando},
  title   = {{R{\'e}nyi} Entanglement Entropy in Fermionic Chains},
  journal = {Int. J. Geom. Methods Mod. Phys.},
  volume  = {12},
  number  = {06},
  pages   = {1560002},
  year    = {2015},
  doi     = {10.1142/S0219887815600026}
}

@article{Ares2015b,
  author  = {Ares, Filiberto and Esteve, Jos{\'e} G. and Falceto, Fernando and de Queiroz, Am{\'i}lcar R.},
  title   = {Entanglement in Fermionic Chains with Finite-Range Coupling and Broken Symmetries},
  journal = {Phys. Rev. A},
  volume  = {92},
  pages   = {042334},
  year    = {2015},
  doi     = {10.1103/PhysRevA.92.042334},
  eprint  = {1506.06665},
  archivePrefix = {arXiv},
  primaryClass  = {quant-ph}
}

@article{Ares2018,
  author  = {Ares, Filiberto and Esteve, Jos{\'e} G. and Falceto, Fernando and de Queiroz, Am{\'i}lcar R.},
  title   = {Entanglement Entropy in the Long-Range {Kitaev} Chain},
  journal = {Phys. Rev. A},
  volume  = {97},
  pages   = {062301},
  year    = {2018},
  doi     = {10.1103/PhysRevA.97.062301},
  eprint  = {1801.07043},
  archivePrefix = {arXiv},
  primaryClass  = {quant-ph}
}

@article{Solfanelli2023,
  author  = {Solfanelli, Andrea and Ruffo, Stefano and Succi, Sauro and Defenu, Nicol{\`o}},
  title   = {Logarithmic, Fractal and Volume-Law Entanglement in a {Kitaev} Chain with Long-Range Hopping and Pairing},
  journal = {Journal of High Energy Physics},
  volume  = {05},
  pages   = {066},
  year    = {2023},
  doi     = {10.1007/JHEP05(2023)066},
  eprint  = {2301.13231},
  archivePrefix = {arXiv},
  primaryClass  = {quant-ph}
}

@article{BaratoSeifert2015,
  author  = {Barato, Andre C. and Seifert, Udo},
  title   = {Thermodynamic Uncertainty Relation for Biomolecular Processes},
  journal = {Phys. Rev. Lett.},
  volume  = {114},
  pages   = {158101},
  year    = {2015},
  doi     = {10.1103/PhysRevLett.114.158101},
  eprint  = {1502.05944},
  archivePrefix = {arXiv},
  primaryClass  = {cond-mat.stat-mech}
}

@article{HasegawaVanVu2019,
  author  = {Hasegawa, Yoshihiko and Van Vu, Tan},
  title   = {Fluctuation Theorem Uncertainty Relation},
  journal = {Phys. Rev. Lett.},
  volume  = {123},
  pages   = {110602},
  year    = {2019},
  doi     = {10.1103/PhysRevLett.123.110602},
  eprint  = {1902.06376},
  archivePrefix = {arXiv},
  primaryClass  = {cond-mat.stat-mech}
}

@article{GrossWitten1980,
  author  = {Gross, David J. and Witten, Edward},
  title   = {Possible third-order phase transition in the large-{$N$} lattice gauge theory},
  journal = {Physical Review D},
  volume  = {21},
  number  = {2},
  pages   = {446--453},
  year    = {1980},
  doi     = {10.1103/PhysRevD.21.446}
}

@article{Wadia1980,
  author  = {Wadia, Spenta R.},
  title   = {{$N = \infty$} phase transition in a class of exactly soluble model lattice gauge theories},
  journal = {Physics Letters B},
  volume  = {93},
  number  = {4},
  pages   = {403--410},
  year    = {1980},
  doi     = {10.1016/0370-2693(80)90353-6}
}

@article{TracyWidom1994,
  author  = {Tracy, Craig A. and Widom, Harold},
  title   = {Level-spacing distributions and the {Airy} kernel},
  journal = {Communications in Mathematical Physics},
  volume  = {159},
  number  = {1},
  pages   = {151--174},
  year    = {1994},
  doi     = {10.1007/BF02100489}
}

\clearpage

\begin{figure}[p]
  \centering

  \begin{minipage}[t]{0.49\textwidth}
    \subfloat[3-mode: Histogram\label{fig:3m-hist}]{%
      \includegraphics[width=\linewidth]{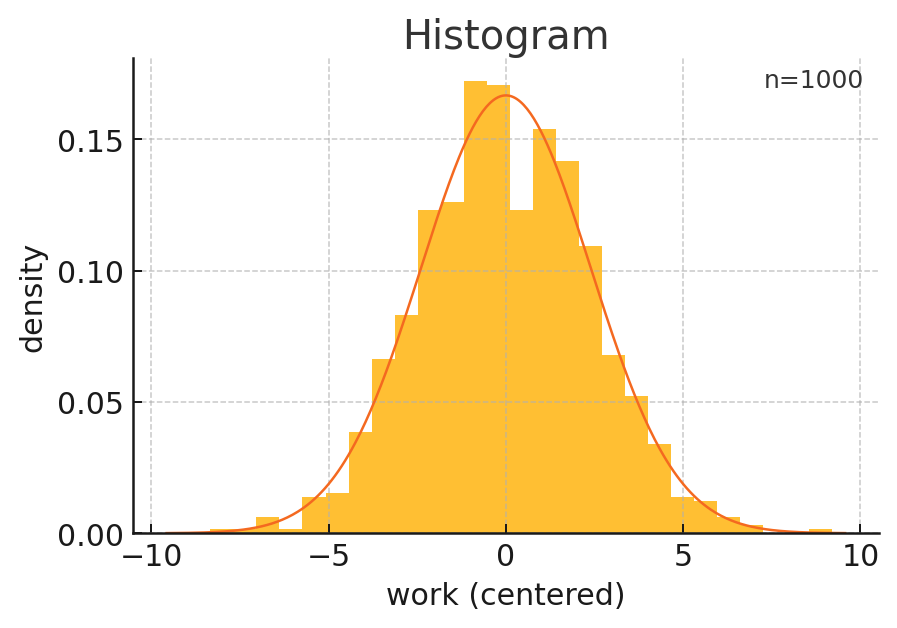}}\\[2pt]
    \subfloat[3-mode: Q--Q\label{fig:3m-qq}]{%
      \includegraphics[width=\linewidth]{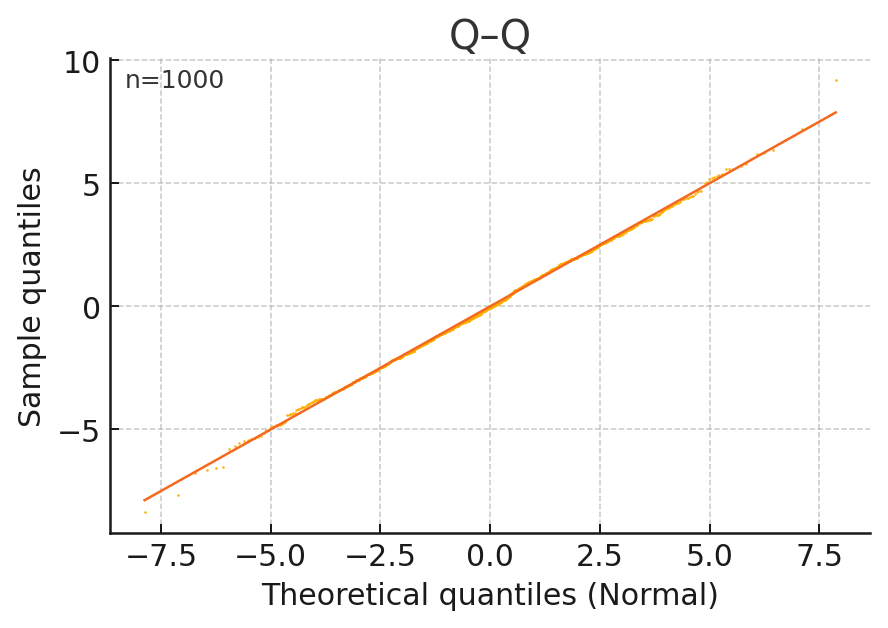}}\\[2pt]
    \subfloat[3-mode: Detrended Q--Q\label{fig:3m-qq-det}]{%
      \includegraphics[width=\linewidth]{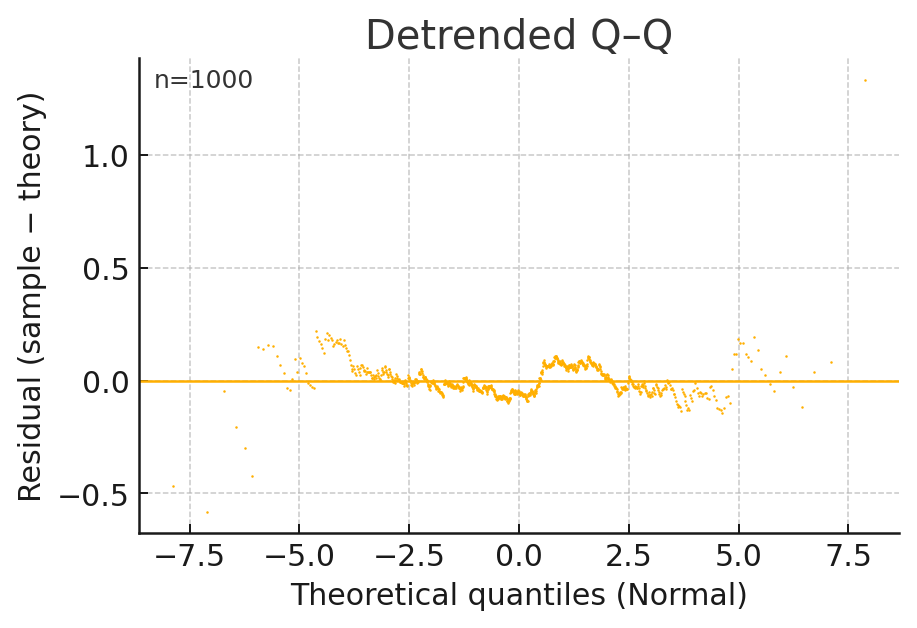}}
  \end{minipage}\hfill
  \begin{minipage}[t]{0.49\textwidth}
    \subfloat[5-mode even: Histogram\label{fig:5m-hist}]{%
      \includegraphics[width=\linewidth]{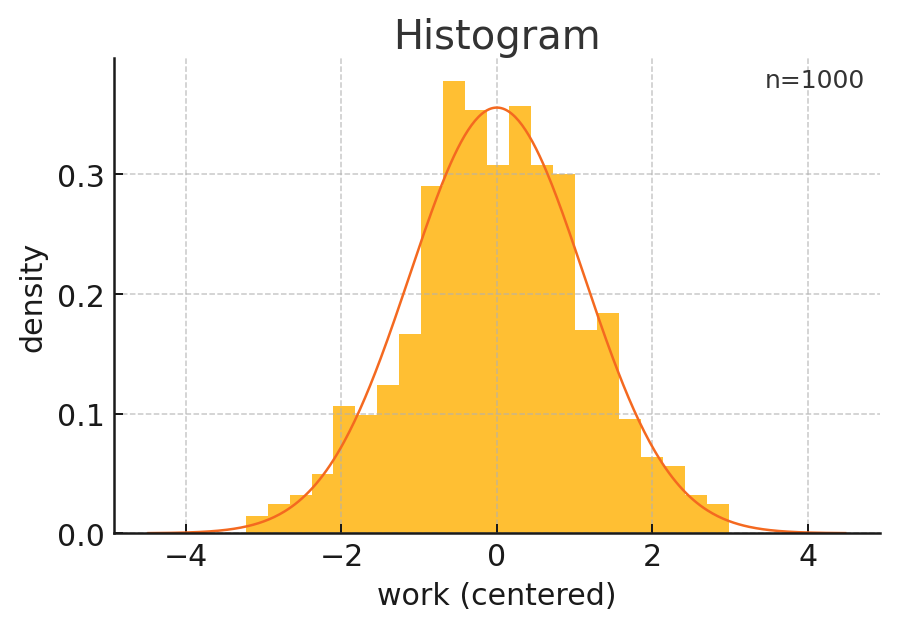}}\\[2pt]
    \subfloat[5-mode even: Q--Q\label{fig:5m-qq}]{%
      \includegraphics[width=\linewidth]{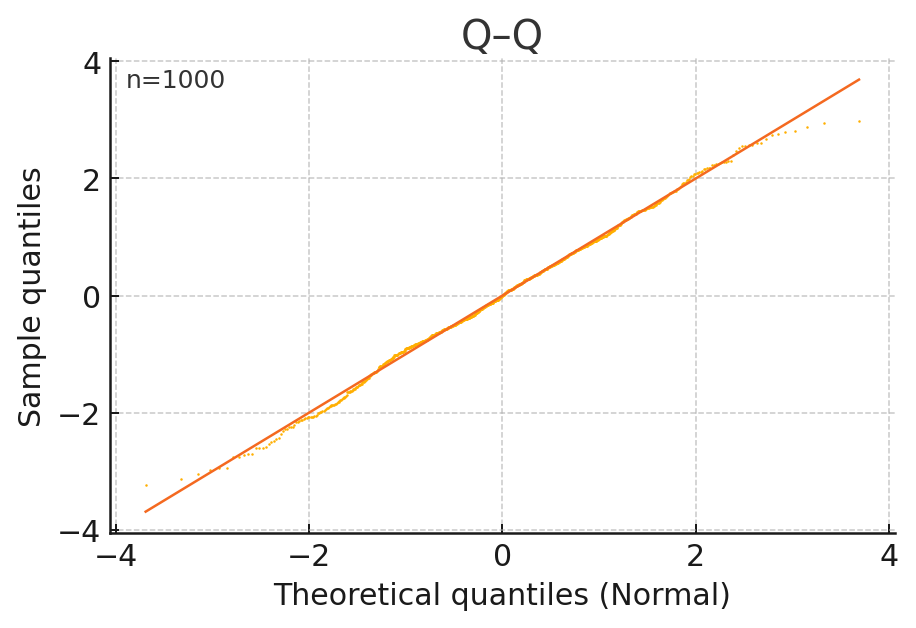}}\\[2pt]
    \subfloat[5-mode even: Detrended Q--Q\label{fig:5m-qq-det}]{%
      \includegraphics[width=\linewidth]{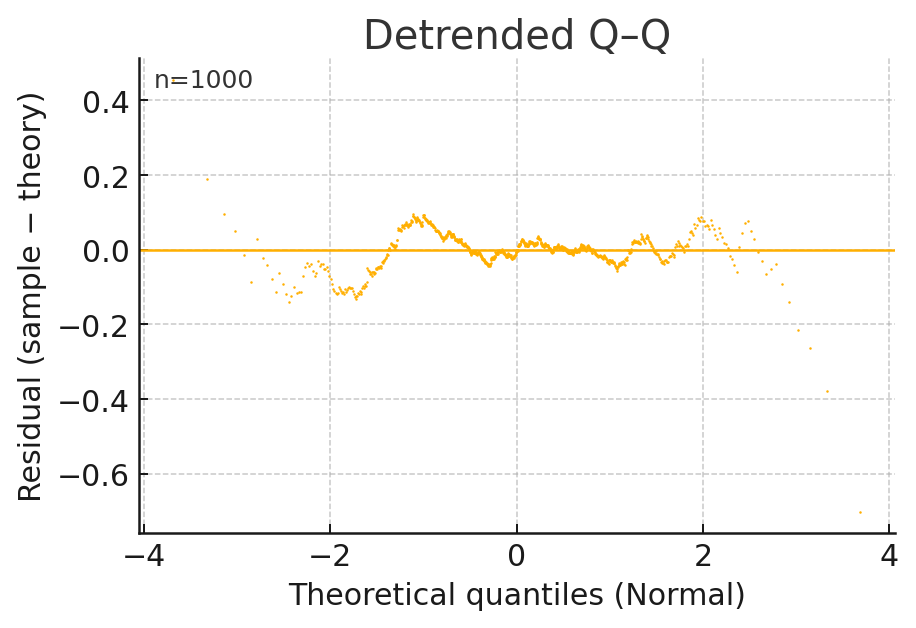}}
  \end{minipage}

  \caption{Side-by-side diagnostics for two families at \(N=80\), \(n=1000\) samples.
  Left: 3-mode model \(\mathcal M=\{1,2,3\}\) with \(\alpha_1=1.0,\alpha_2=0.7,\alpha_3=0.5\).
  Right: 5-mode even model \(\mathcal M=\{2,4,6,8,10\}\), \(a_m=e^{-0.4m}\).
  Each column shows the histogram with Gaussian overlay, the Q--Q plot (points + \(y=x\)),
  and the detrended Q--Q (residuals vs theory).}
  \label{fig:grid-3m-5m}
\end{figure}

\begin{figure}[p]
  \centering
  \begin{minipage}[t]{0.49\textwidth}
    \subfloat[Power-law: Histogram\label{fig:pl-hist}]{%
      \includegraphics[width=\linewidth]{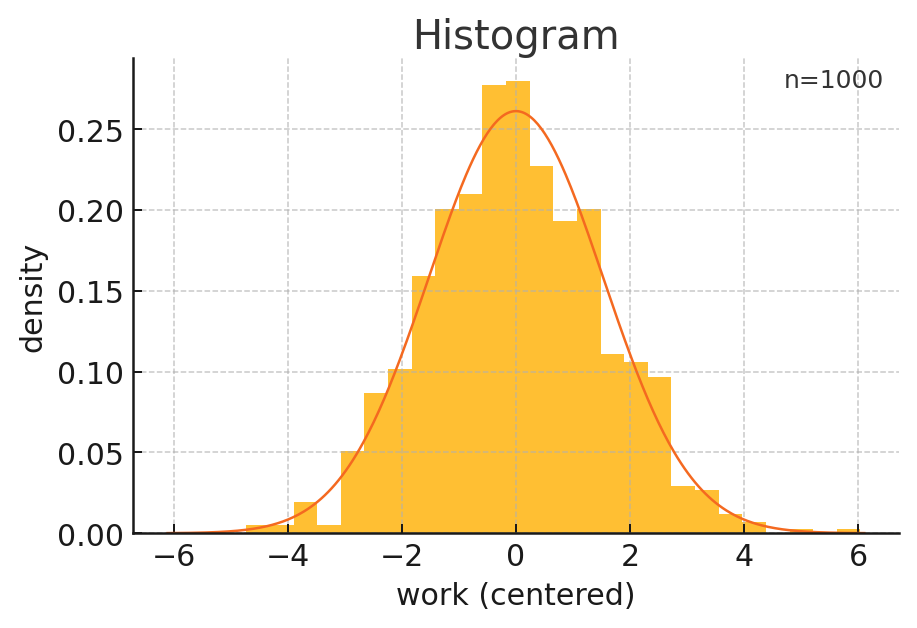}}\\[2pt]
    \subfloat[Power-law: Q--Q\label{fig:pl-qq}]{%
      \includegraphics[width=\linewidth]{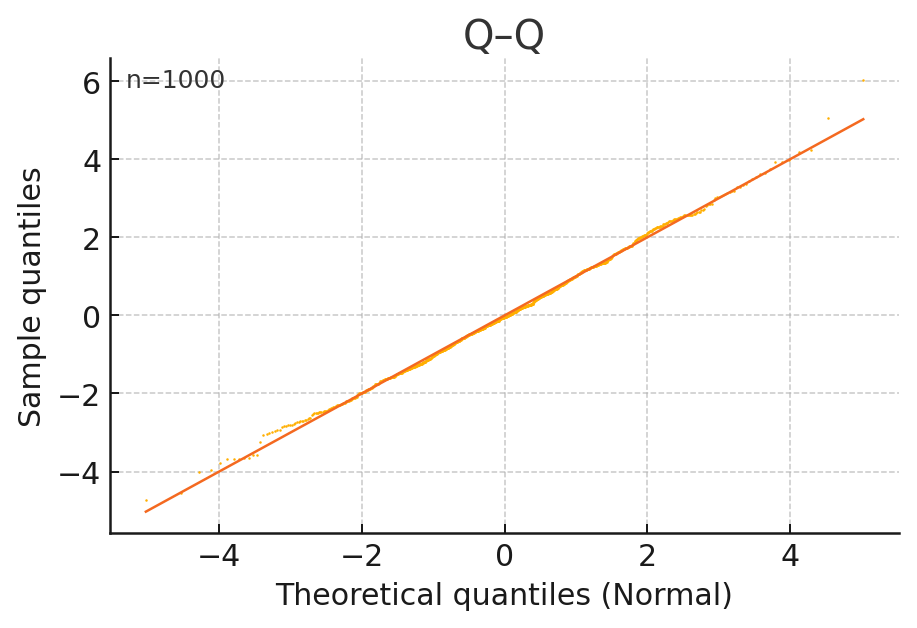}}\\[2pt]
    \subfloat[Power-law: Detrended Q--Q\label{fig:pl-qq-det}]{%
      \includegraphics[width=\linewidth]{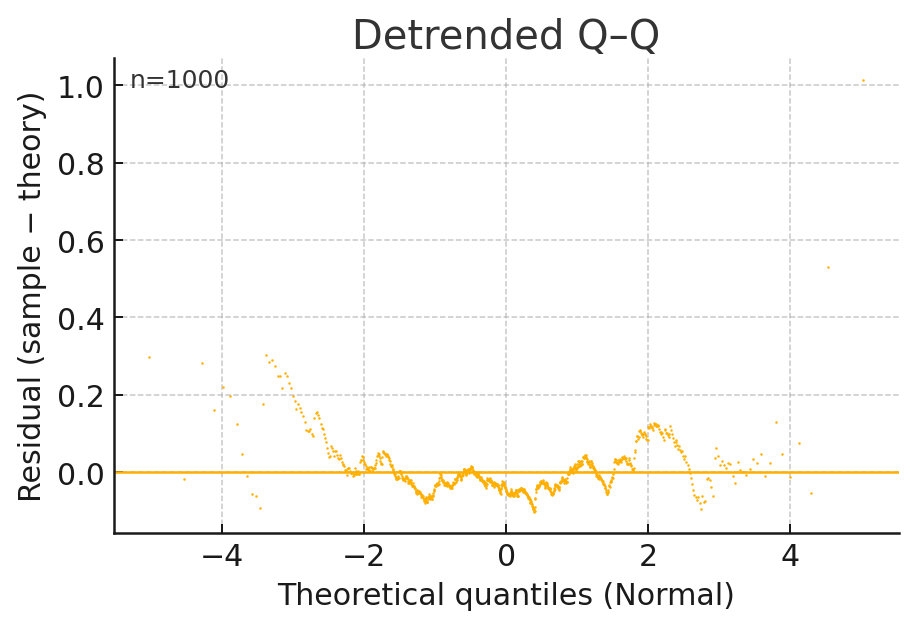}}
  \end{minipage}\hfill
  \begin{minipage}[t]{0.49\textwidth}
    \subfloat[Exponential even: Histogram\label{fig:exp-hist}]{%
      \includegraphics[width=\linewidth]{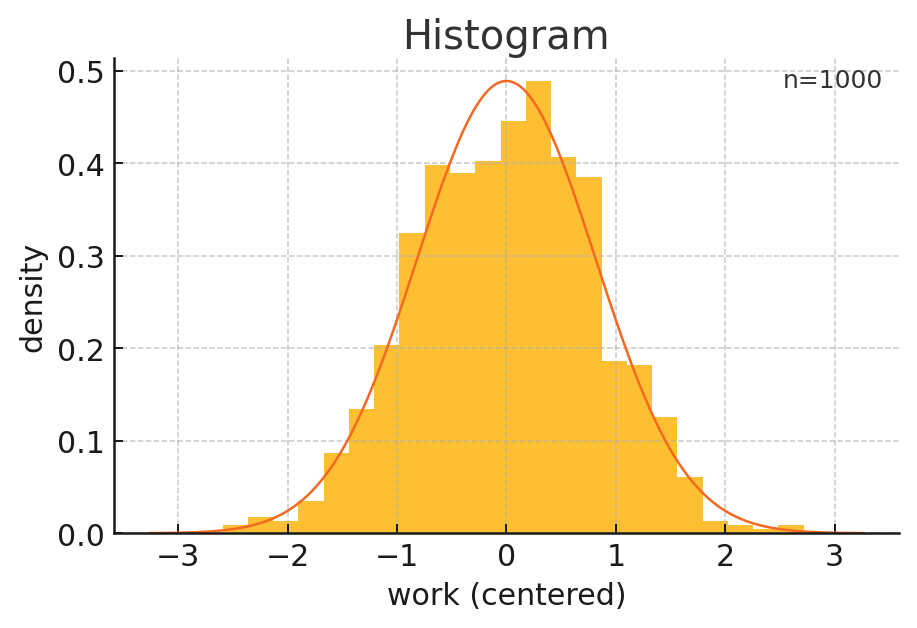}}\\[2pt]
    \subfloat[Exponential even: Q--Q\label{fig:exp-qq}]{%
      \includegraphics[width=\linewidth]{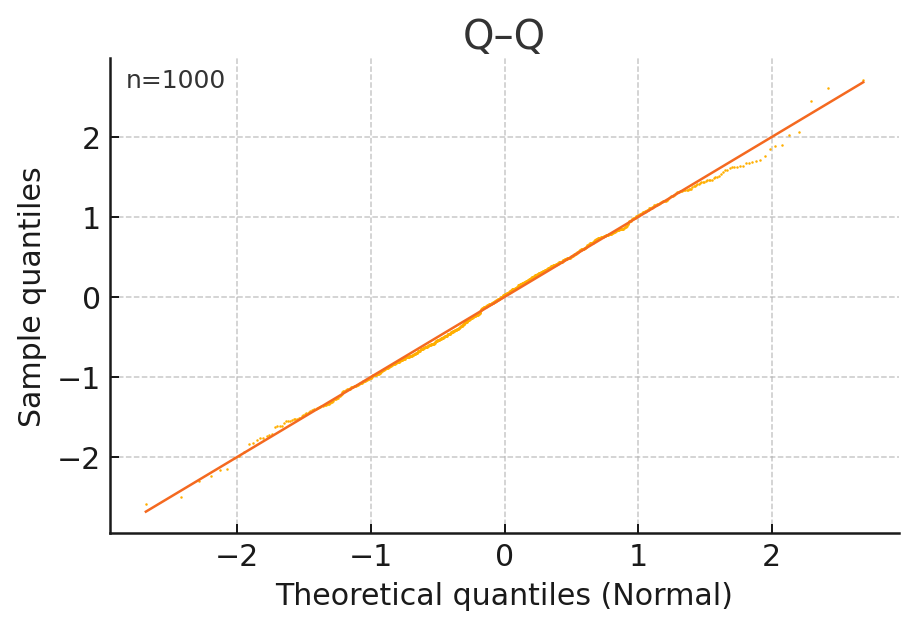}}\\[2pt]
    \subfloat[Exponential even: Detrended Q--Q\label{fig:exp-qq-det}]{%
      \includegraphics[width=\linewidth]{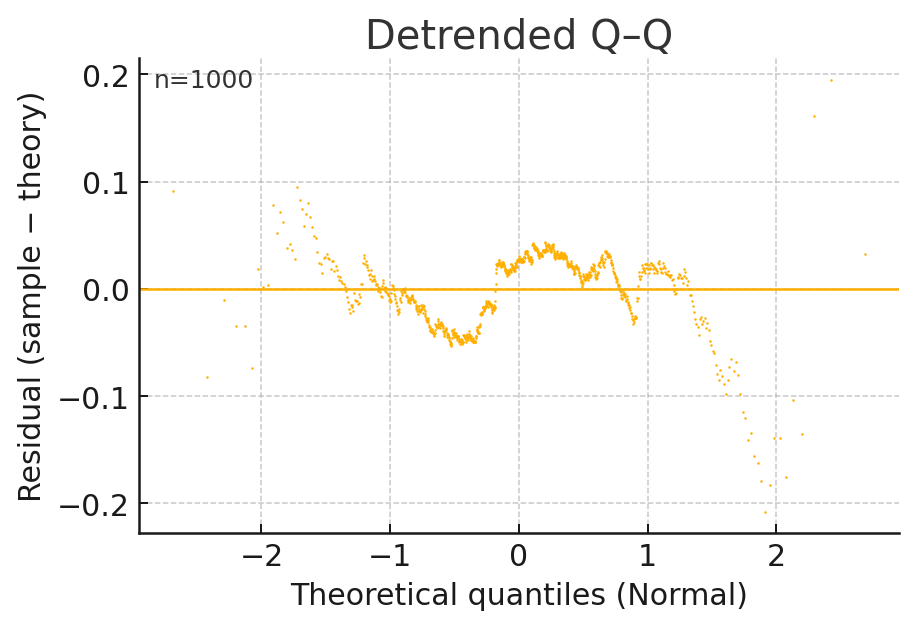}}
  \end{minipage}
  \caption{Diagnostics for power-law (left; \(a_m=m^{-2}\)) and exponential even (right; \(a_m=e^{-0.5m}\)) families at \(N=80\), \(n=1000\) samples.}
  \label{fig:grid-pl-exp}
\end{figure}

\begin{figure}[p]
  \centering

  \begin{minipage}[t]{0.485\textwidth}
    \subfloat[Hist.\ (Haar vs sur.)\label{fig:S1-3mode-hist}]{%
      \includegraphics[width=\linewidth]{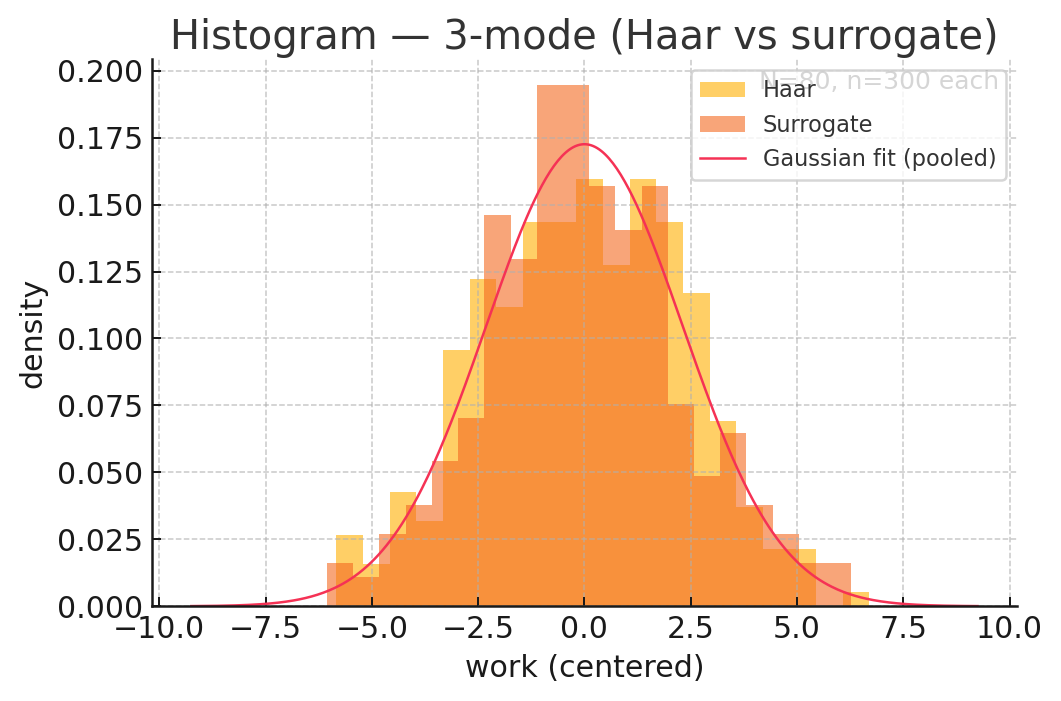}}\\[2pt]
    \subfloat[Q--Q (Haar)\label{fig:S1-3mode-qq-haar}]{%
      \includegraphics[width=\linewidth]{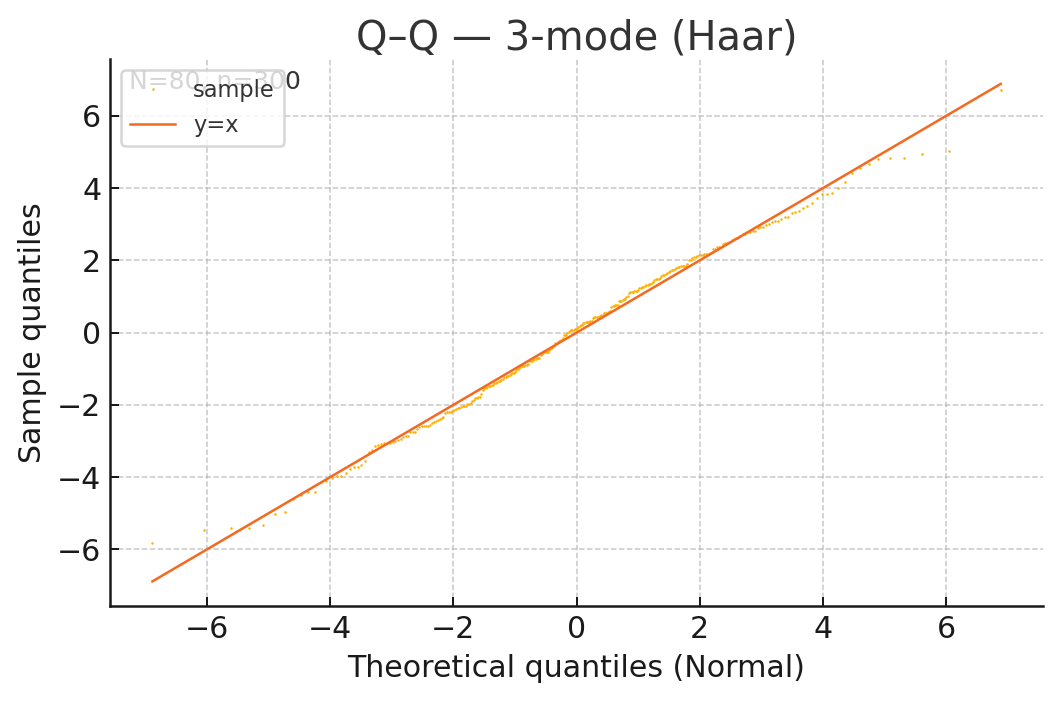}}\\[2pt]
    \subfloat[Q--Q (Sur.)\label{fig:S1-3mode-qq-sur}]{%
      \includegraphics[width=\linewidth]{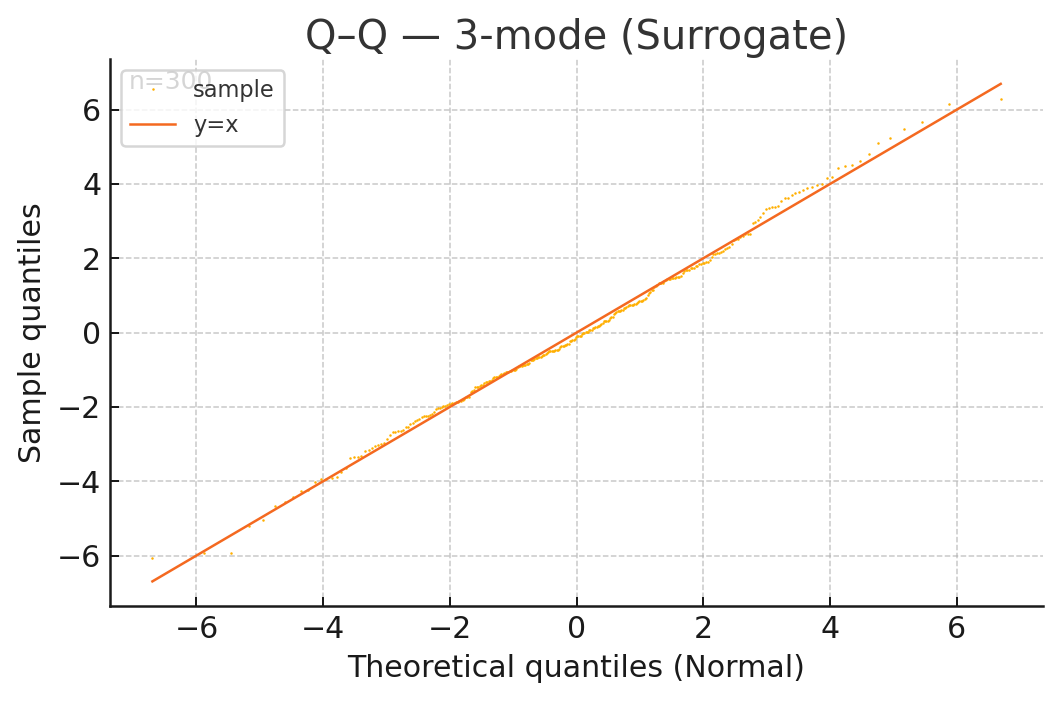}}
  \end{minipage}\hfill%
  \begin{minipage}[t]{0.485\textwidth}
    \subfloat[Hist.\ (Haar vs sur.)\label{fig:S1-5mode-hist}]{%
      \includegraphics[width=\linewidth]{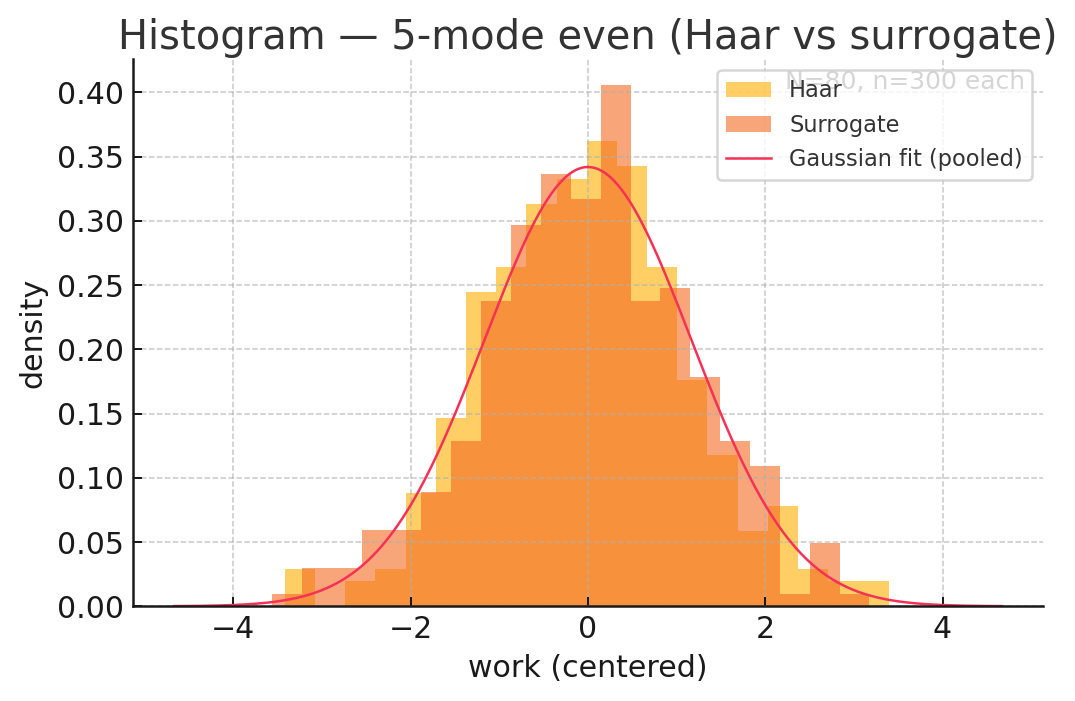}}\\[2pt]
    \subfloat[Q--Q (Haar)\label{fig:S1-5mode-qq-haar}]{%
      \includegraphics[width=\linewidth]{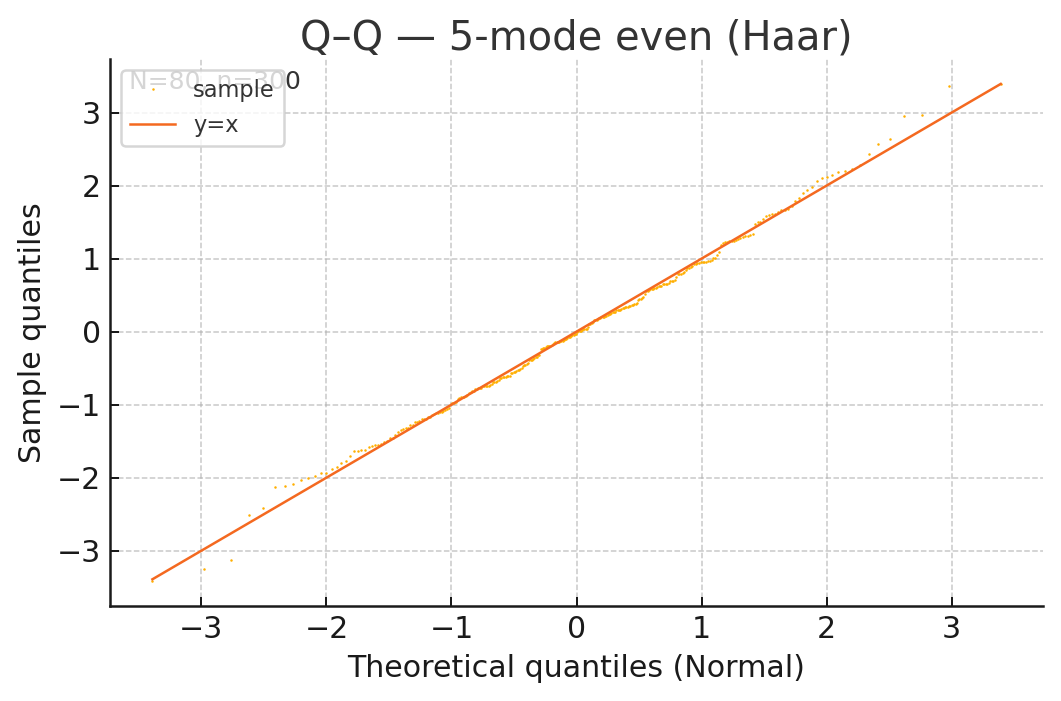}}\\[2pt]
    \subfloat[Q--Q (Sur.)\label{fig:S1-5mode-qq-sur}]{%
      \includegraphics[width=\linewidth]{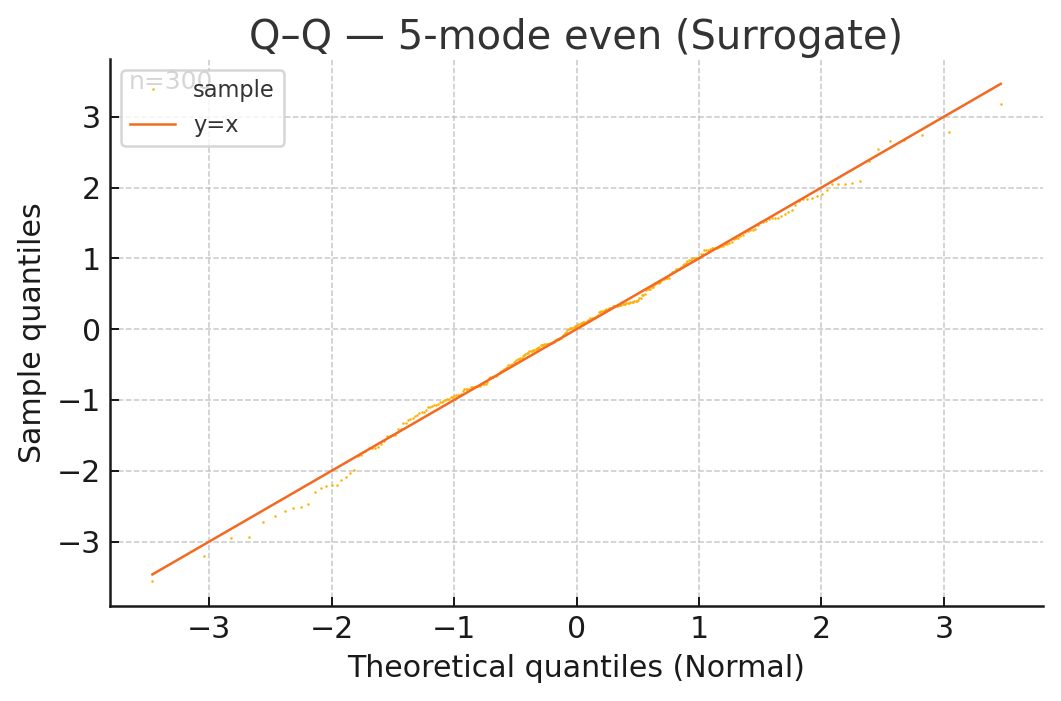}}
  \end{minipage}

  \caption{Surrogate vs exact Haar at \(N=80\), \(n=300\) per sample.
  Left: 3‑mode \(\mathcal{M}=\{1,2,3\}\), \(\alpha_1=1.0,\alpha_2=0.7,\alpha_3=0.5\).
  Right: 5‑mode even \(\mathcal{M}=\{2,4,6,8,10\}\), \(a_m=e^{-0.4m}\).
  Rows: (top) histogram overlays (Freedman–Diaconis bins; pooled Gaussian fit),
  (middle) Q--Q (Haar), (bottom) Q--Q (surrogate).}
  \label{fig:grid-haar-surrogate}
\end{figure}

\begin{figure}[p]
  \centering
  \subfloat[Variance vs \(N\).]{%
    \includegraphics[width=0.46\textwidth]{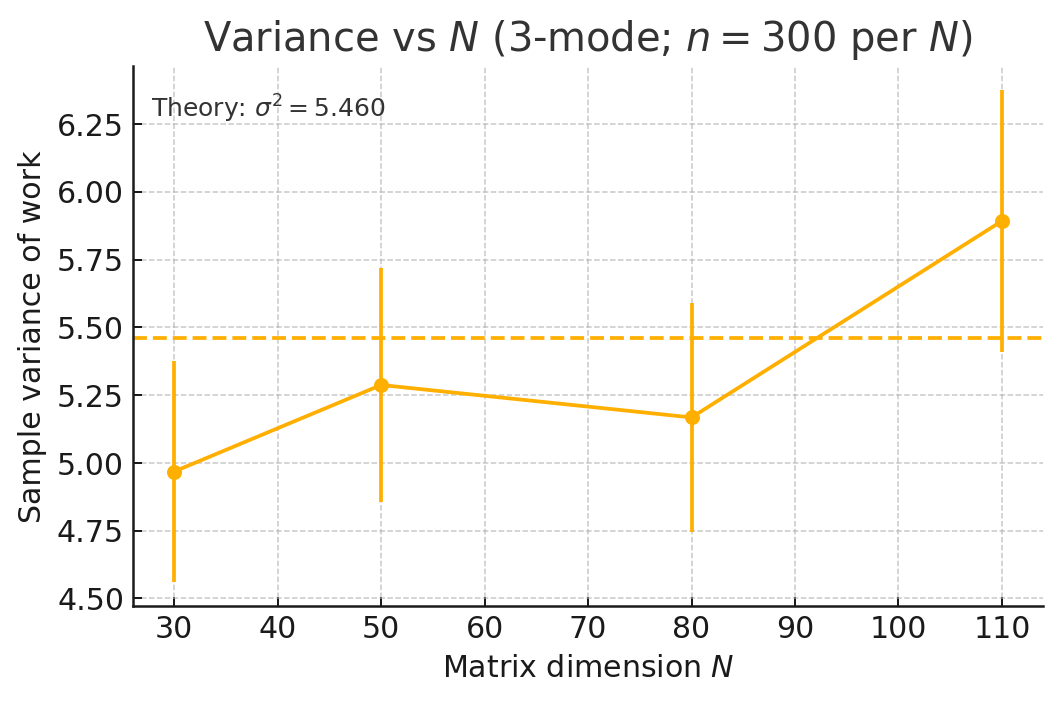}%
    \label{fig:var-vs-N}}
  \hfill
  \subfloat[Excess kurtosis vs \(N\).]{%
    \includegraphics[width=0.46\textwidth]{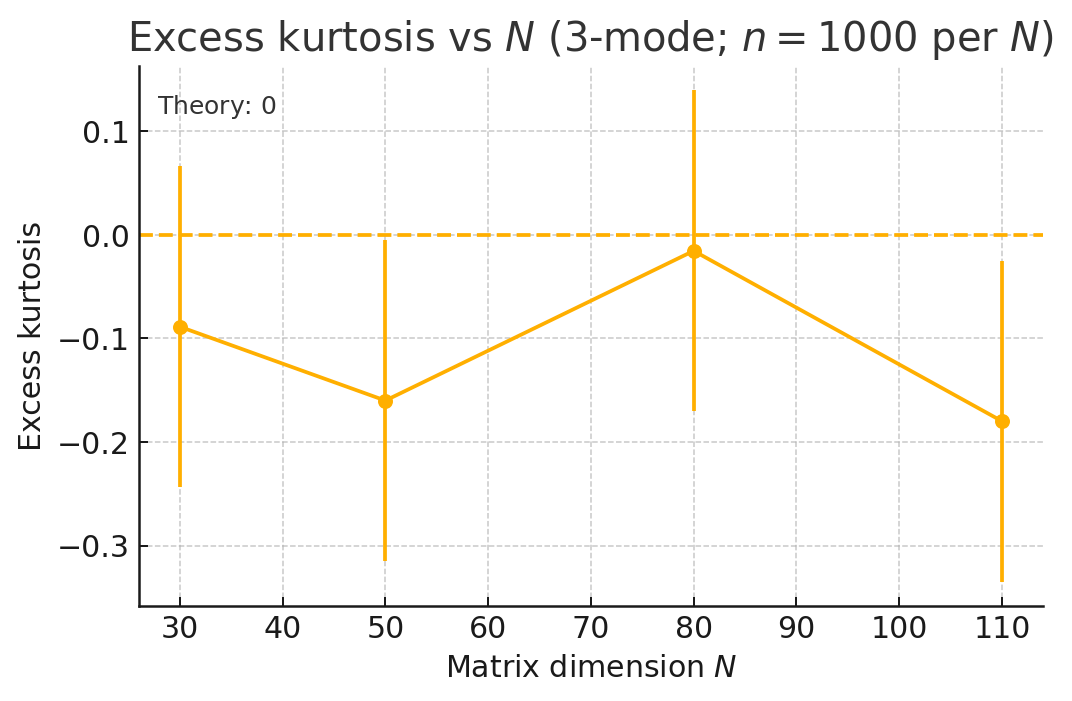}%
    \label{fig:kurtosis-vs-N}}
  \\[4pt]
  \subfloat[Long-range proxy: kurtosis of block means.]{%
    \includegraphics[width=0.46\textwidth]{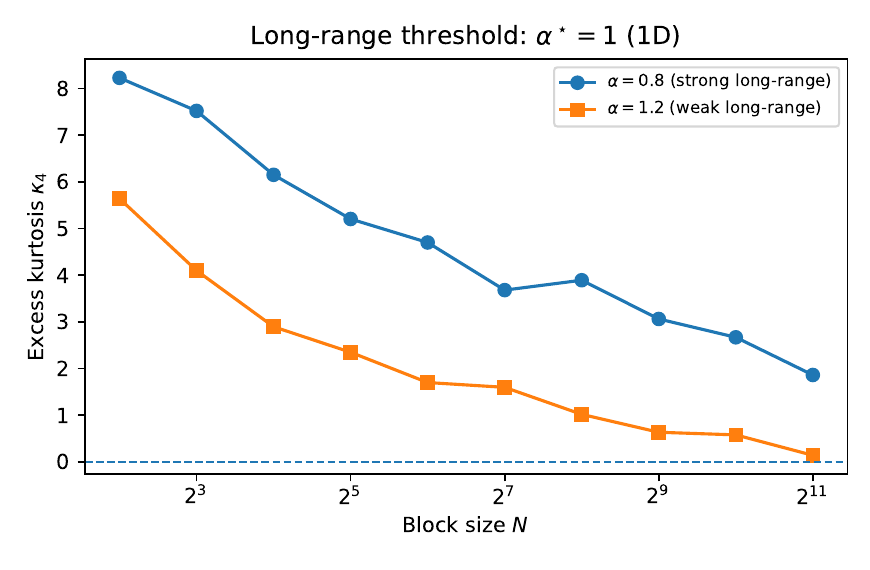}%
    \label{fig:longrange-threshold}}
  \hfill
  \subfloat[CUE check: variance scaling for $a_r=r^{-\alpha}$.]{%
    \includegraphics[width=0.46\textwidth]{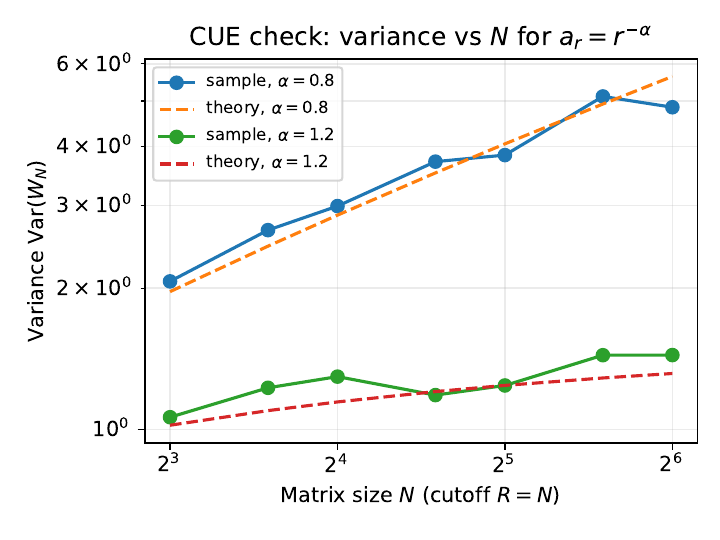}%
    \label{fig:cue-variance}}
  \caption{Scaling diagnostics.
  (a)~Sample variance vs \(N\) for the 3-mode model \(\mathcal M=\{1,2,3\}\), \(\alpha_1=1.0\), \(\alpha_2=0.7\), \(\alpha_3=0.5\), with analytic prediction (dashed) and SE bars \(\hat\sigma^2\sqrt{2/(n-1)}\) (\(n=300\) per \(N\)).
  (b)~Excess kurtosis vs \(N\) for the same model, with error bars \(\sqrt{24/n}\) (\(n=1000\) per \(N\)); dashed line marks the Gaussian limit.
  (c)~Pedagogical CLT diagnostic for the long-range threshold \(\alpha^*=d=1\): excess kurtosis of standardized block means of a non-Gaussian field \(y=x^2-1\) built from fractional Gaussian noise with connected correlations \(\langle y_0 y_r\rangle_c\sim r^{-\alpha}\).
  For \(\alpha=0.8<1\) (strong long-range, \(H=0.8>3/4\)) the kurtosis remains elevated at all block sizes, signalling persistent non-Gaussianity (Dobrushin--Major regime); for \(\alpha=1.2>1\) (weak long-range, \(H=0.7<3/4\)) it decays toward zero, confirming CLT convergence.
  (d)~Direct CUE check of the Szeg\H{o} sum mechanism: for Haar $U(N)$, define the linear statistic
  $W_N=\sum_{r=1}^{N} a_r\,\Re\,\mathrm{Tr}\,U^r$ with $a_r=r^{-\alpha}$ (cutoff $R=N$).
  Points show the Monte--Carlo sample variance (here \(n=400\) per \(N\)); dashed curves show the exact prediction
  $\mathrm{Var}(W_N)=\tfrac12\sum_{r=1}^{N} r\,a_r^2=\tfrac12\sum_{r=1}^{N} r^{1-2\alpha}$ from unitary orthogonality~\cite{DiaconisShahshahani1994}.
  For \(\alpha>1\) the variance saturates (finite Szeg\H{o} sum), while for \(\alpha<1\) it grows with $N$ (divergent Szeg\H{o} sum), matching the threshold discussed in Sec.~\ref{sec:fermions}.}
  \label{fig:variance-kurtosis}
\end{figure}

\begin{figure}[p]
  \centering
  \subfloat[\(\Re\,\mathrm{Tr}\,U\) vs \(\Re\,\mathrm{Tr}\,U^{2}\).]{%
    \includegraphics[width=0.46\textwidth]{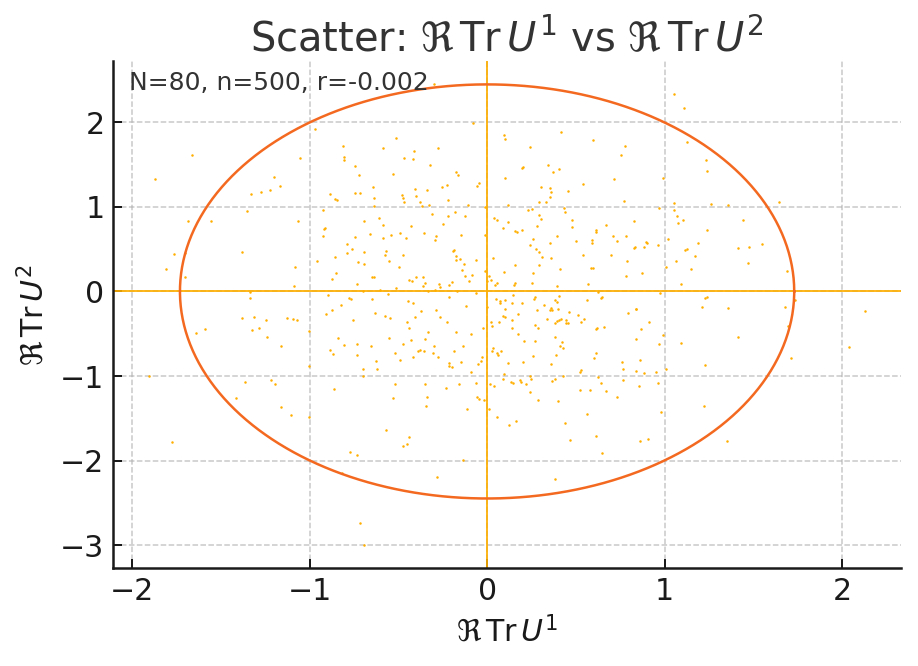}%
    \label{fig:scatter12}}
  \hfill
  \subfloat[\(\Re\,\mathrm{Tr}\,U\) vs \(\Re\,\mathrm{Tr}\,U^{3}\).]{%
    \includegraphics[width=0.46\textwidth]{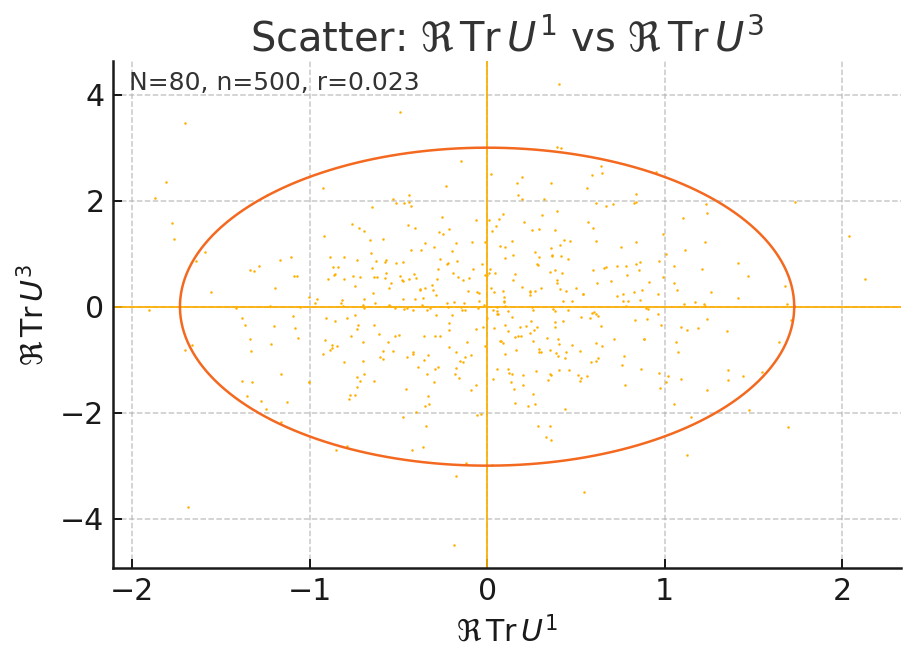}%
    \label{fig:scatter13}}\\[3pt]
  \subfloat[Standardized: \(m=1\) vs \(m=2\).]{%
    \includegraphics[width=0.46\textwidth]{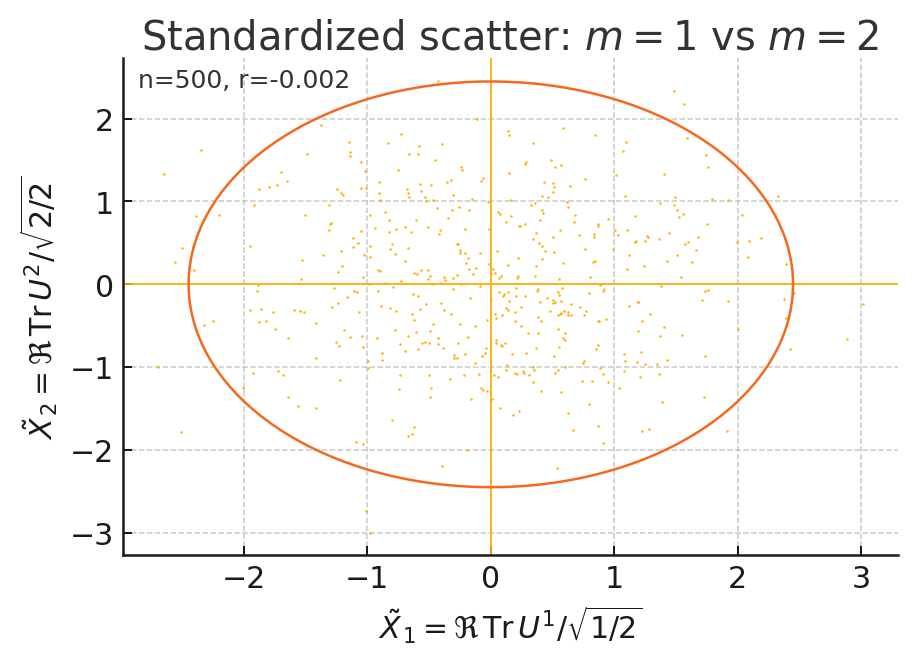}%
    \label{fig:scatter12-std}}
  \hfill
  \subfloat[Standardized: \(m=1\) vs \(m=3\).]{%
    \includegraphics[width=0.46\textwidth]{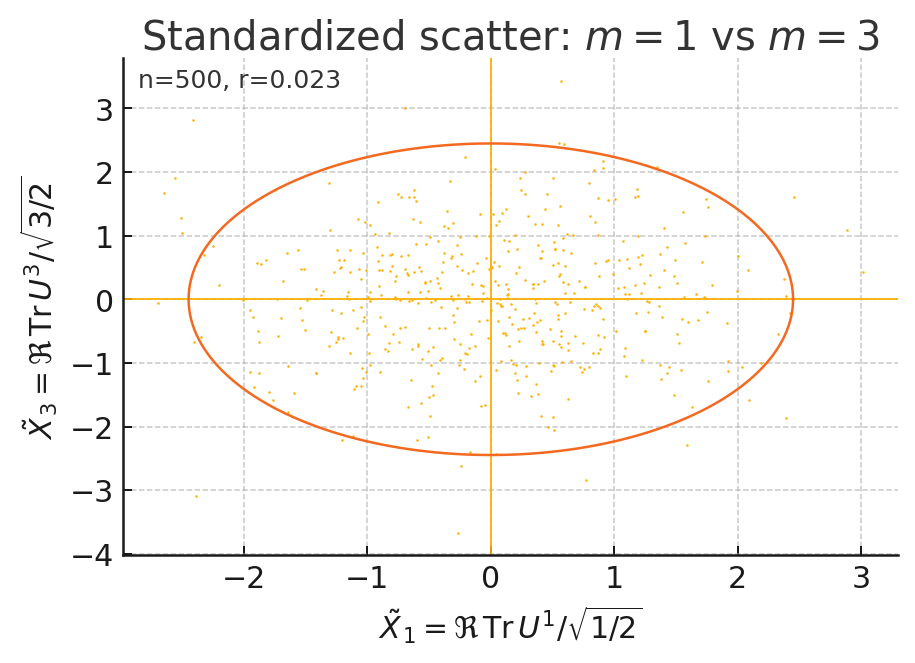}%
    \label{fig:scatter13-std}}
  \caption{Scatter diagnostics with theory 95\% contours (ellipse/circle). Top: raw coordinates show the \((\sqrt{m/2})\) scale on each axis; bottom: standardized coordinates divide by \(\sqrt{m/2}\) so the 95\% contour is a circle, making different \((m,n)\) pairs directly comparable. Annotations in the panels report the sample correlation \(r\) (theory: \(r=0\)); here \(n=500\) samples.}
  \label{fig:scatter-grid}
\end{figure}



\end{document}